\documentclass[12pt,draftclsnofoot, onecolumn]{IEEEtran}

\usepackage{algorithm}
\usepackage{algorithmic}
\usepackage{amsmath}
\usepackage{amssymb}
\usepackage{latexsym}
\usepackage{multirow}
\usepackage{epsfig}
\usepackage{graphics}
\usepackage{graphicx}
\usepackage{mathrsfs}
\usepackage{subfigure}
\usepackage{bbding}
\usepackage{epic}
\usepackage{curves}
\usepackage{stfloats}
\usepackage{epstopdf}
\usepackage{cite}
\usepackage{stfloats}
\usepackage{color}

\newcommand{\beq}{\begin{equation}}
	\newcommand{\enq}{\end{equation}}
\newcommand{\ben}{\begin{eqnarray}}
	\newcommand{\enn}{\end{eqnarray}}
\newcommand{\bei}{\begin{itemize}}
	\newcommand{\eni}{\end{itemize}}

\newcommand{\bm}[1]{\mbox{\boldmath{$#1$}}}
\newtheorem{proposition}{Proposition}
\newtheorem{lemma}{Lemma}
\newtheorem{remark}{Remark}
\newtheorem{proof}{Proof}
\makeatletter
\newcommand{\figcaption}{\def\@captype{figure}\caption}
\newcommand{\tabcaption}{\def\@captype{table}\caption}
\makeatother
\date{}

\begin{document}
	"This work has been submitted to the IEEE for possible publication.  Copyright may be transferred without notice, after which this version may no longer be accessible. "
\title{ 	
	 \hfill{\em\small{}}\\
	 \LARGE 
	 Lifetime Maximization for UAV-Enabled Cognitive-NOMA IoT Networks: Joint Location, Power, and Decoding Order Optimization
%
%

 \author{Na~Tang, Hongying~Tang, Tho~Le-Ngoc (Life Fellow, IEEE), Baoqing~Li, Xiaobing~Yuan}
 \thanks{N. Tang is with the Science and Technology on Microsystem Laboratory,	Shanghai Institute of Microsystem and Information Technology, Chinese	Academy of Sciences, Shanghai 200050, China, and also with the University of Chinese Academy of Sciences, Beijing 100049, China (e-mail: tangna@mail.sim.ac.cn).}
 \thanks{Hongying Tang, Baoqing Li, and Xiaobing Yuan are with the Science	and Technology on Microsystem Laboratory, Shanghai Institute of Microsystem and Information Technology, Chinese Academy of Sciences, Shanghai 200050, China (e-mail: tanghy@mail.sim.ac.cn; sinoiot@mail.sim.ac.cn;	sinowsn@mail.sim.ac.cn).}
  \thanks{Tho~Le-Ngoc is with the Department of Electrical and Computer Engineering, McGill University, Montreal, QC H3A 0G4, Canada (email: tho.le-ngoc@mcgill.ca).}}

\maketitle        

\begin{abstract}
This paper investigates a cognitive unmanned aerial vehicle (UAV) enabled Internet of Things (IoT) network, where secondary/cognitive IoT devices upload their data to the UAV hub following a non-orthogonal multiple access (NOMA) protocol in the spectrum of the primary network.
We aim to maximize the minimum lifetime of IoT devices by jointly optimizing the UAV location,  transmit power, and decoding order subject to interference-power constraints in presence of the imperfect channel state information (CSI).
To solve the formulated non-convex mixed-integer programming problem, 
we first jointly optimize the UAV location and transmit power for a given decoding order and obtain the globally optimal solution with the assistance of Lagrange duality and then obtain the best decoding order by exhaustive search,  which is applicable to relatively small-scale scenarios.
For large-scale scenarios, we propose a low-complexity sub-optimal algorithm by transforming the original problem into a more tractable equivalent form and applying the successive convex approximation (SCA) technique and penalty function method.
Numerical results demonstrate that the proposed design significantly outperforms the benchmark schemes.
\end{abstract}


%
\IEEEpeerreviewmaketitle

\section{Introduction}
\IEEEPARstart{I}{nternet} of Things (IoT) has been widely used in various applications, including smart home,
health-care, manufacturing, soil monitoring, and fish farming \cite{6740844,islam2020development}. A large number of intelligent devices capable of processing and sensing are anticipated to connect to the data center,
which imposes critical challenges to spectrum resources. Due to size and cost constraints, many
IoT devices carry drastically limited battery capacity and their lifetime becomes a bottleneck to ensure long-term performance in scenarios
where replacing the device or battery is difficult and cost-prohibitive. For example, in precision
aquaculture systems, IoT devices are used to gather environmental data (such as oxygen and
temperature) to help reduce disease-related fish mortality. In this case, the available device energy
is limited by its battery. Once the battery runs out, it is unable to support stable and reliable
communications to transmit the collected data and even renders the system useless. Therefore, it
is of utmost importance to efficiently utilize spectrum and battery resources to prolong the lifetime
of IoT networks.

Cognitive radio (CR) technology has emerged as a prominent solution to spectrum scarcity issues.
The basic idea of the underlay CR paradigm is to allow secondary users to concurrently access the licensed band of primary networks under interference constraints at primary users, thereby enhancing the spectrum utilization \cite{huang2015green,7559749}.
However, CR cannot satisfy the demands of simultaneous transmission among massive devices, which arises as one of the biggest hurdles for IoT networks.

Non-orthogonal multiple access (NOMA) is proposed as an effective method to improve network access capabilities.
The main principle of power domain NOMA is to achieve time/frequency spectrum sharing and power domain multiplexing among multiple devices with the aid of superposition coding (SC) at the transmitter and successive interference cancellation (SIC) at the receiver.
In this way, NOMA provides several advantages over orthogonal multiple access (OMA), such as high spectral efficiency, massive connectivity support, and low transmission latency \cite{maraqa2020survey,8357810}.

Unmanned aerial vehicles (UAVs) have been proposed as powerful aerial base stations (BSs) to serve ground users from the sky.
With fully controllable mobility, employing UAVs to collect data from IoT devices can significantly shorten the communication distance and provide better channel quality.
Compared with the terrestrial BS, UAVs have a much higher probability to connect IoT devices through line-of-sight (LoS) links, thus providing more reliable and energy-efficient wireless transmission.
Thanks to these advantages, UAVs are perfectly matched with IoT networks, especially for devices with limited energy or installed in inaccessible areas.
Therefore, the integration of CR, NOMA, and UAV can potentially enhance the efficiency of the spectrum
and battery resources in IoT networks.
This motivates us to investigate the cognitive NOMA transmission in UAV-enabled IoT networks.

\subsection{Related Work}
UAV-assisted communications have been widely investigated in the literature, in which UAVs are mainly employed as BSs \cite{8438896,zhan2017energy,9043712} or relays \cite{9522072,9524328}.
For instance, a UAV-enabled orthogonal frequency division multiple access (OFDMA) network is investigated in \cite{8438896} to maximize the throughput by optimizing the UAV trajectory and resource allocation.
A time-division multiple access (TDMA) like wake-up scheduling scheme is proposed in \cite{zhan2017energy} with the goal of minimizing the maximum energy consumption of sensors.
For a UAV-enabled full-duplex relaying system, the work in \cite{9522072} studies the joint design of bandwidth, power allocation, and the UAV trajectory
to maximize the number of served IoT devices.

Extensive efforts have been devoted to studying the application of CR in UAV communications.
There has been a proliferation of researches on employing the UAV to assist CR networks \cite{9070201,9082700,9364745}.
Specifically, the work in \cite{9070201} uses the UAV as a relay to cooperate the communication of secondary users in energy-harvesting-based cognitive networks, while the work in \cite{9364745} deploys the UAV as a friendly jammer to interfere with eavesdroppers.
Moreover, there has been an increasing interest in cognitive UAV-enabled communications \cite{8776639,9233353,9014324}.
In \cite{8776639}, the UAV location/trajectory and transmit power are jointly optimized to maximize the secondary users' rates for quasi-stationary and mobile UAV scenarios, respectively.
Ultra-reliable and low-latency communications are investigated in \cite{9014324} to maximize the minimum finite block-length rate by optimizing the UAV altitude and power allocation.

Exploiting the attractive features of NOMA to further unlock the full potential of UAV communications has attracted great attention.
Recent studies have investigated various objectives that include rate \cite{9126800,9257576,9509753,9194041,8848428}, UAV transmit power\cite{8848428,9113466}, energy efficiency \cite{9411713}, number of served users\cite{8918266}, as well as secure transmission\cite{tang2020cognitive,8988182,9080059}.
Among them, a data collection protocol is proposed in \cite{9126800} to maximize the sum rate of a wireless sensor network by jointly optimizing the UAV placement, sensor grouping, and power control.
The backhaul connectivity of the UAV-BS is considered in \cite{9113466}, where location and resource allocation are jointly optimized to minimize the UAV transmit power.
\cite{8848428} studies the joint optimization problem of location, decoding order, and power allocation with two goals of minimizing the UAV transmit power and maximizing the achievable rate of a user.
Reference \cite{9411713} investigates the joint resource allocation and UAV trajectory optimization for maximizing the total energy efficiency with the quality of service (QoS) requirement.
To maximize the number of users with satisfied QoS, the work in \cite{8918266} studies the UAV location design, admission control, and power allocation for NOMA transmission.
In \cite{8988182}, two schemes are proposed to improve the security of NOMA-UAV networks via power allocation and beamforming optimization, respectively.
In \cite{tang2020cognitive}, the secrecy sum-rate maximization problem is studied for downlink transmission in a UAV-enabled cognitive NOMA network.
However, these works largely ignore the lifetime performance in NOMA-based UAV communications.

There exist quite a few works focusing on lifetime maximization in UAV-enabled communication systems.
For disaster situations, the placement problem of UAVs is studied in \cite{shakhatreh2019uavs} to maximize the minimum device lifetime, where the lifetime is defined as the uplink transmission time until the first wireless device runs out of energy.
The work in \cite{8886053} further optimizes device association, power, bandwidth allocation, and UAV deployment for maximizing the minimum lifetime in a multi-UAV-enabled communication system.
However, \cite{shakhatreh2019uavs} and \cite{8886053} adopt the frequency division multiplexing access (FDMA) technique for uplink transmission and their results cannot be directly applied to NOMA systems, due to different resource allocation mechanisms.
For an ocean monitoring network with NOMA based device-to-UAV transmission, the authors in \cite{ma2021uav} study the multiple access, resource allocation, and UAV deployment issues to maximize the number of data collection cycles supported by the residual energies of underwater sensors, which is related to the above-defined the lifetime.
 However, in their model, the UAV location is first determined by sensors' positions and thus the decoding order is fixed based on channel gains of devices, which is strictly sub-optimal.
Up to now, the joint optimization of UAV location, transmit power, and decoding order for improving the lifetime performance of NOMA networks has not been fully studied.

\subsection{Contributions and Organization}
In this paper, we investigate the cognitive NOMA transmission in UAV-enabled IoT networks, where cognitive IoT devices transmit data to the secondary UAV with NOMA by accessing the spectrum allocated to the primary users.
Our goal is to maximize the minimum lifetime of secondary IoT devices by jointly optimizing the UAV location, transmit power, and decoding order subject to the QoS and interference constraints in the presence of the imperfect channel state information (CSI) between secondary IoT devices and primary users.
The main contributions of this paper can be summarized as follows.
\begin{itemize}
	\item Although the formulated problem is non-convex mixed-integer and challenging to tackle, we solve it optimally by developing the Lagrange-duality-based algorithm. Specifically, for any given decoding permutation order, the optimal power and location are obtained in semi-closed forms. Then the problem resorts to enumerating all the possible decoding orders and determining the best performer that yields the highest lifetime. Note that the Lagrange-duality-based algorithm is applicable for small-scale scenarios in practice and can also serve as a benchmark for evaluating other algorithms.
	\item Furthermore, we also propose a low-complexity iterative algorithm, which can obtain a sub-optimal solution for large-scale scenarios. Specifically, by introducing binary variables to replace the complicated permutation variables, we construct an equivalent but more tractable problem. Then we derive a closed-form solution of the optimal power, by decoupling the mutual relationship between each user in the non-convex QoS constraints. Finally, the SCA technique and penalty function method are applied to approximate the location and decoding order problem into a series of convex problems.  
	\item 	Numerical results demonstrate that the proposed joint optimization of UAV location, transmit power, and decoding order significantly outperforms the benchmark schemes. Besides, compared with the Lagrange-duality-based algorithm that achieves the optimal lifetime, the SCA-based iterative algorithm shows a slight performance degradation, which demonstrates that the SCA-based iterative algorithm is sub-optimal.
\end{itemize}

	It is worth mentioning that different from previous works that only adopt binary variables \cite{zhang2019optimal,lu2020uav,tang2020cognitive} or permutation variables \cite{8848428,8685130,8918266,xu2020joint,nguyen2018novel} to formulate the decoding order in NOMA transmission,   our work unifies these two common decoding order formulations for the first time by proving the equivalence between the original problem and the newly constructed problem.

	The rest of this paper is organized as follows. Section \ref{sec: lifetime} introduces the system model and problem formulation for the cognitive NOMA  transmission in UAV-enabled  IoT networks.  Section \ref{sec: sub-optimal} and Section \ref{sec: optimal} propose optimal and sub-optimal solutions to the formulated problem, respectively.
	In Section \ref{sec: result}, numerical results are presented to validate the effectiveness of the proposed joint design. Finally, we conclude the paper in Section \ref{sec: conclusion}.
	
	\emph{Notations}:
	In this paper, we use italic letters to denote scalars and bold lowercase letters to denote vectors.
	$\mathbb{R}^{M \times 1}$ is the
	space of M-dimensional real-valued vectors. For a scalar $x$, $ |x| $ is the absolute value.
	For a vector $\mathbf{y}$, $ ||\mathbf y|| $ is the Euclidean norm, and $\mathbf y^T$ denotes the transpose.


%
\begin{figure}[!t]
	\centering
	\includegraphics[width=3in]{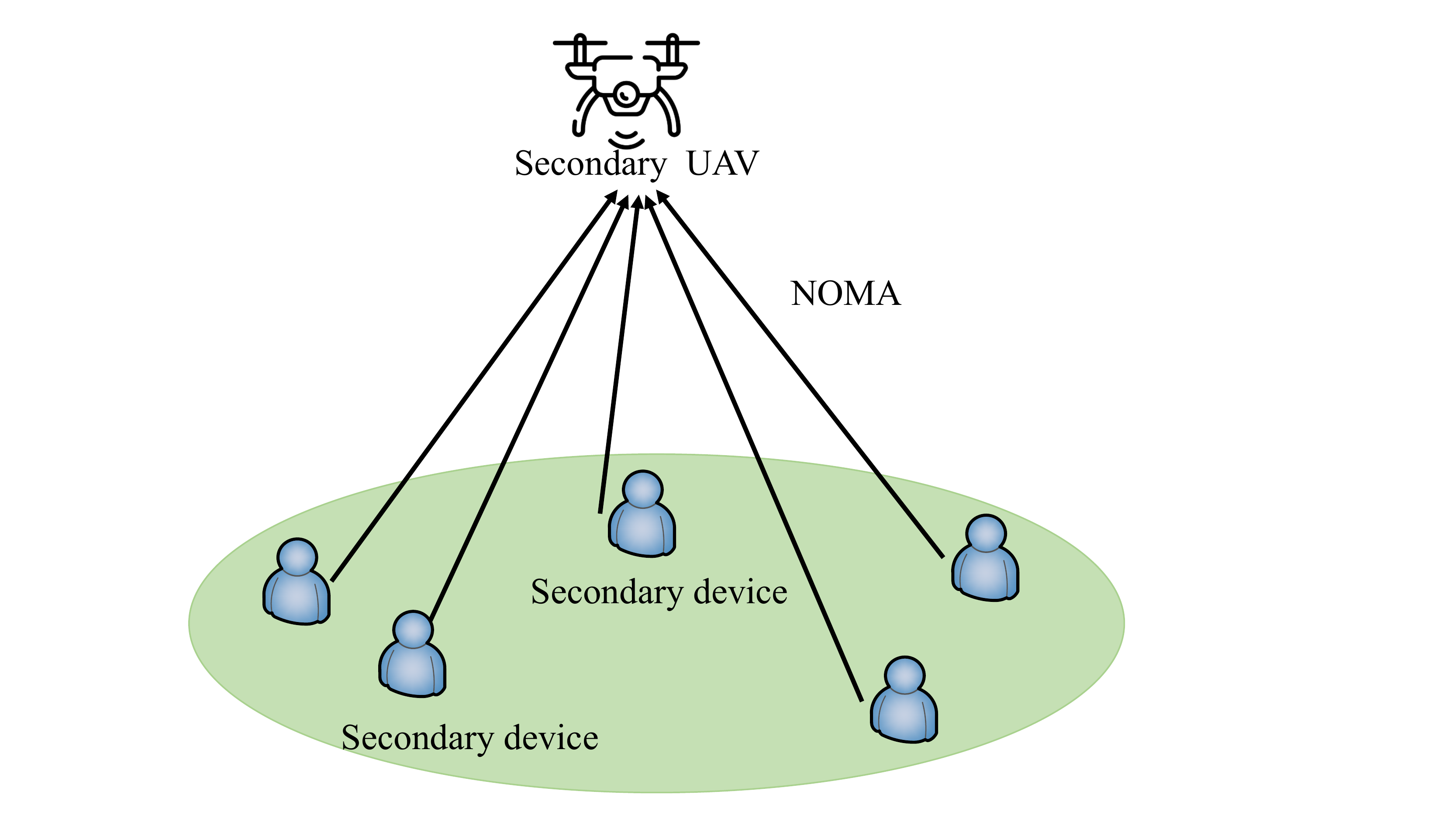}
	\caption{ A UAV-enabled secondary IoT network.}\label{fig: systemmodelifetime}
\end{figure}

\section{System Model and Problem Formulation}\label{sec: lifetime}
\subsection{System Model}\label{sec: lifetimesystemmodel}
As shown in Fig. \ref{fig: systemmodelifetime}, we consider a spectrum sharing scenario for a UAV-enabled IoT network, where
$K$ secondary IoT devices upload data to the secondary UAV with NOMA over the licensed spectrum resource of the primary users.
Let $\mathcal{K} = \{1,2,\cdots,K\}$ denote the set of ground IoT devices.
Without loss of generality, we adopt a three-dimensional ($3$D) Cartesian coordinate system such that the horizontal coordinate of ground IoT device $k$ is denoted as $\mathbf{w}_k \in \mathbb{R}^{2 \times 1}$, $k \in \mathcal{K}$.
Furthermore, we assume that the UAV is deployed at a constant altitude $H$ with horizontal coordinate $\mathbf{q}\in \mathbb{R}^{2 \times 1}$.
In this paper, we consider offline optimization by assuming
that the UAV can get the specific location information of IoT devices in advance.


We assume that IoT devices are located outdoors in open areas (e.g., river or farm) with limited blockage or scattering and the UAV-to-device communication channels are mainly dominated by the LoS link  \cite{8438896,9485092}.
Therefore, the UAV can easily obtain CSI of devices based on location information \cite{9043712}.
In addition, we assume that the Doppler effect due to UAV mobility can be fully compensated.
Following the free-space path loss model, the channel power gain from device $k$ to the UAV can be expressed as
\begin{align}
	{h_{k}}= \rho_{0} d_{k}^{-2}=\frac{\rho_0 }{||\mathbf{q}- \mathbf{w}_{k}||^2+H^2}, \label{equ: powergain1}
\end{align}
where $\rho_0$ denotes the channel power gain at the reference distance $d_0 = 1 $ m.

In uplink NOMA transmission, the UAV receives a superimposed signal of all devices and performs SIC to decode each signal.
The device's signal with the highest channel gain (referred to as the strongest signal) is first decoded and then is subtracted from the aggregate signal.
The same procedure is followed by the next strongest signal \cite{ali2016dynamic}.
Consequently, the device with the highest channel gain (referred to as the strongest device) experiences interference from all devices, and the device with the lowest channel gain (referred to as the weakest device) effectively experiences interference-free transmission.
For ease of exposition, denote the SIC decoding order at the UAV as $\bm{\pi} = \{\pi(1), \pi(2),\cdots,\pi(K)\}$, where $\pi(m)=k$ means that the signal of device $k$  is the $m$-th one to be decoded.
As a result, the above constraints can be equivalently modeled as
\begin{align}
	&	\{\pi(1),\pi(2),\cdots, \pi(K)\} \in \bm{\Psi},\label{equ: temp06}\\
	&	 	h_{\pi(1)}\geq \cdots \geq h_{\pi(K)},\label{equ: channelsort0}
\end{align}
where $\bm{\Psi}$ denotes the set of all the $K!$ possible permutations from $1$ to $K$.
Moreover, the signal from device $\pi(n)$ with $n>m$ is treated as interference when decoding the signal from device $\pi(m)$.
Thus, the signal-to-interference-plus-noise ratio (SINR) of device $\pi(m)$ is expressed as
\begin{align}
	\gamma_{\pi(m)} &= \frac{p_{\pi(m)}{h}_{\pi(m)}}{\sum_{n=m+1}^{K}p_{\pi(n)}{h}_{\pi(n)}+\sigma^{2}},
\end{align}
where  $p_\pi(m)$ denotes the transmit power of device $\pi(m)$ and $\sigma^2$ is the variance of Gaussian noise.
Define the device lifetime as the uplink transmission time before its battery energy is drained \cite{shakhatreh2019uavs}. Then the lifetime of device $\pi(m)$, denoted by $\tau_{\pi(m)}$ in second (s), can be expressed as
\begin{equation}
	\tau_{\pi(m)} = \frac{E_{\pi(m)}}{p_{\pi(m)}+P_c},\label{UNequ: tauquality}
\end{equation}
where $E_{\pi(m)}$ denotes the battery energy of device $\pi(m)$ and $P_c$ denotes the circuit power consumption.

We assume a flat-fading channel between IoT device $k$ and the primary user, where the channel coefficient is denoted as $c_{\text{p},k}$, which is an independent and identically distributed (i.i.d) zero mean circularly symmetric complex Gaussian (ZMCSCG) random variable, i.e., $CN(0, 1)$ \cite{5336868,9014324}.
We further assume that IoT devices only have imperfect CSI of $c_{\text{p},k}$ and carry out minimum mean square error (MMSE) estimation of $c_{\text{p},k}$, which is written as
$	c_{\text{p,k}} = \breve{c}_{\text{p,k}}+\Delta c_{\text{p,k}}$,
with $\breve{c}_{\text{p,k}}$ and $\Delta c_{\text{p,k}}$ denoting the MMSE estimation of $c_{\text{p,k}}$ and corresponding estimation error, respectively.
According to the property of MMSE estimation, $\breve{c}_{\text{p,k}}$ and $\Delta c_{\text{p,k}}$ are uncorrelated and are ZMCSCG distributed with variance $1-\sigma_e^2$ and $\sigma_e^2$, respectively.
The channel power gains are denoted as 	$z_{k}$, $\tilde{z}_{k}$ and $\Delta z_{k}$ with $z=|c|^2$, and the corresponding probability density functions can be, respectively, expressed as
\begin{align}
	&f\left(z_{ k}\right)=\exp \left(-z_{ k}\right), f\left(\tilde{z}_{\text{p,k}}\right)=\frac{1}{1-\varepsilon_{p}^{2}} \exp \left(-\frac{\tilde{z}_{\text{p,k}}}{1-\varepsilon_{p}^{2}}\right),f\left(\Delta z_{\text{p},k}\right)=\frac{1}{\varepsilon_{p}^{2}} \exp \left(-\frac{\Delta z_{\text{p},k}}{\varepsilon_{p}^{2}}\right).
\end{align}
Since IoT devices only obtain the imperfect estimation of $\tilde{z}_{\text{p},k}$, it cannot strictly ensure that the interference peak power at the primary user is below the threshold $I_{th}$.
Assume that the transmission of IoT devices is allowed to exceed the peak threshold with a certain probability $\varrho$, i.e.,
\begin{align}
\operatorname{Pr}\{p_{\pi(m)}z_{\pi(m)} \geq I_{th}\} \leq \varrho, \forall m. \label{equ: threshold1}
\end{align}

\subsection{Problem Formulation }\label{sec: lifetimeproblem}
Let  $\mathbf p=\{p_{\pi(m)}, m \in \mathcal{K}\}$.
By taking into account the QoS and interference constraints, our objective is to maximize the minimum lifetime of IoT devices via jointly optimizing the UAV location $\mathbf{q}$, transmit power $\mathbf p$, and decoding order $\bm{\pi}$.
Mathematically, the optimization problem can be formulated as
\begin{subequations}\label{equ: problem31}
	\begin{align}
		\mathop {\max }\limits_{\eta, \mathbf{q},\mathbf{p},\bm{\pi}} & \quad \eta\\	
		\text{s.t.} & \quad \frac{E_{\pi(m)}}{p_{\pi(m)}+P_c} \geq \eta,\forall m, \label{equ: tau2}\\
		&\quad 0 \leq p_{\pi(m)} \leq P_{\max},\forall m,\label{equ: power1}\\	
		&\quad \log_2(1+\gamma_{\pi(m)}) \geq r^*,\forall m,\label{equ: qos0}\\
		&\quad \eqref{equ: temp06},\eqref{equ: channelsort0},\eqref{equ: threshold1},
	\end{align}
\end{subequations}
where $P_{\max}$ is the maximum transmit power of devices and $r^*$ denotes the QoS requirement for all devices.
Due to the existence of non-convex constraints \eqref{equ: temp06}, \eqref{equ: channelsort0}, \eqref{equ: threshold1}, \eqref{equ: tau2}, and \eqref{equ: qos0} as well as the permutation variable \bm{\pi} that lacks closed-form expression, problem \eqref{equ: problem31} is a mixed-integer non-convex optimization problem, which is generally hard to solve.

To facilitate the processing of problem \eqref{equ: problem31}, we transform it into a more tractable form by exploiting the hidden convexity of the non-convex constraints in the following. 
First, based on the results in \cite{5336868}, the interference constraints \eqref{equ: threshold1} can be guaranteed by a sufficient instantaneous power constraint, which is expressed as
\begin{align}
	p_{\pi(m)} \leq \frac{I_{t h}}{z_{ \pi(m)}-\varepsilon_{p}^{2} \ln \varrho}.\label{equ: threshold2}
\end{align}
From \eqref{equ: powergain1}, constraint \eqref{equ: channelsort0} can be equivalently rewritten as
\begin{align}
	d_{\pi(1)}\leq d_{\pi(2)} \leq \cdots \leq d_{\pi(K)}.\label{equ: channelsort}
\end{align}
Moreover, we introduce the slack variable $\zeta =\frac{1}{\eta}$, by which problem \eqref{equ: problem31} can be transformed into
\begin{subequations}\label{equ: problem33}
	\begin{align}
		\mathop {\min }\limits_{{\zeta,\mathbf{q},\mathbf{p},\bm{\pi}}} &\quad \zeta\\
		\text{s.t.}& \quad p_{\pi(m)}+P_c \leq \zeta E_{\pi(m)},\forall m,\label{equ: sslacm3}\\
		&\quad 	 0 \leq p_{\pi(m)} \leq \tilde{P}_{\pi(m)} ,\forall m,\label{equ: power33}\\
		&\quad \eqref{equ: qos0},\eqref{equ: channelsort},
	\end{align}
\end{subequations}
where $\tilde{P}_{\pi(m)}\triangleq \min\{P_{\max}, \frac{I_{t h}}{z_{ \pi(m)}-\varepsilon_{p}^{2} \ln \varrho}\}$ that will be referred to as the allowable power in the following content.
Note that the non-convex constraints \eqref{equ: tau2} have been converted into linear constraints \eqref{equ: sslacm3}, which greatly reduces the computational complexity. 
Problem \eqref{equ: problem33} is still mixed-integer non-convex due to non-convex constraints \eqref{equ: qos0} and \eqref{equ: channelsort} as well as the permutation variable \bm{\pi}.
In general, it is quite challenging to solve such a problem optimally.
In the next two sections, we propose an optimal algorithm as well as a low-complexity sub-optimal algorithm for solving problem \eqref{equ: problem33} by applying the Lagrange duality, SCA technique and penalty method, respectively.

\section{Global Optimization Solution}\label{sec: optimal}
\label{sec: multicase}
In this section, we provide a globally optimal solution of problem \eqref{equ: problem33} by proposing the Lagrange-duality-based algorithm.
Specifically, for given decoding order \bm{\pi}, we first derive a closed-form expression of the optimal power $\mathbf{ p}$ and thereby reduce the joint design to the location $\mathbf{q}$ optimization problem that is convex. Then, the optimal $\mathbf{q}$ is obtained in a semi-closed form.
Finally, by searching all possible decoding orders and picking out the best performer, we obtain the global optimum to problem \eqref{equ: problem33}.

\subsection{Joint Location and Transmit Power Optimization}\label{sec: multidevicesA}
For any given decoding order $\bm{\pi}$, the UAV location and transmit power $\{\mathbf{q},\mathbf{p}\}$ can be optimized by solving the following problem.
\begin{subequations}\label{equ: problem34}
	\begin{align}
		\mathop {\min }\limits_{\zeta,\mathbf{q},\mathbf{p}} &\quad \zeta\\
		\text{s.t.}& \quad p_{\pi(m)}+P_c \leq \zeta E_{\pi(m)},\label{equ: sslack4}\forall m,\\
		&\quad 	 0 \leq p_{\pi(m)} \leq \tilde{P}_{\pi(m)} ,\forall m,\label{equ: power35}\\
		&\quad \eqref{equ: qos0}, \eqref{equ: channelsort}.
	\end{align}
\end{subequations}
Problem (29) is a non-convex optimization problem due to non-convex constraints \eqref{equ: qos0} and \eqref{equ: channelsort}. 
By exploiting problem structure and the hidden convexity, we propose an efficient algorithm based on Lagrange duality to determine the optimal $\{\mathbf{q}, \mathbf{p}\}$.
Specifically, note that constraints \eqref{equ: channelsort} can be equivalently rewritten as
\begin{align}
	||\mathbf{q}- \mathbf{w}_{\pi(1)}||^2	\leq ||\mathbf{q}- \mathbf{w}_{\pi(2)}||^2 \leq\cdots \leq ||\mathbf{q}- \mathbf{w}_{\pi(K)}||^2.\label{equ: dist32}
\end{align}
 However, due to $||\mathbf{q}- \mathbf{w}_{\pi(m)}||^2- ||\mathbf{q}- \mathbf{w}_{\pi(m+1)}||^2 = 2(\mathbf{w}_{\pi(m+1)}-\mathbf{w}_{\pi(m)})^T\mathbf{q}-||\mathbf{w}_{\pi(m+1)}||^2+||\mathbf{w}_{\pi(m)}||^2$, constraints \eqref{equ: dist32} are further transformed into
\begin{align}
	2(\mathbf{w}_{\pi(m+1)}-\mathbf{w}_{\pi(m)})^T\mathbf{q}&\leq||\mathbf{w}_{\pi(m+1)}||^2-||\mathbf{w}_{\pi(m)}||^2,\forall m=1,\cdots,K-1,\label{equ: serialdist2}
\end{align}
which are linear with respect to $\mathbf{q}$.
Moreover, constraints \eqref{equ: qos0} can be reformulated as
\begin{align}
	p_{\pi(m)}{{h}}_{\pi(m)} \geq \left(2^{r^*}-1\right)
	\left(\sum_{n= m+1}^K p_{\pi(n)}{h}_{\pi(n)}+\sigma^2\right),\forall m.\label{equ: qos1}
\end{align}
Thus, problem \eqref{equ: problem34} is equivalent to
\begin{subequations}\label{equ: problem35}
	\begin{align}
		\mathop {\min }\limits_{{\zeta,\mathbf{q},\mathbf{p}}} &\quad \zeta\label{equ: obj11} \\
		\text{s.t.}& \quad p_{\pi(m)}+P_c \leq \zeta E_{\pi(m)},\label{equ: sslack5}\forall m,\\
		&\quad \eqref{equ: power35}, \eqref{equ: serialdist2},\eqref{equ: qos1}.
	\end{align}
\end{subequations}

Combining \eqref{equ: obj11} and \eqref{equ: sslack5}, it can be verified that we have $\zeta(\mathbf{p}) =\mathop {\max }_{m}\{\frac{p_{\pi(m)}+P_c}{E_{\pi(m)}}\}$ at the optimal solution, since otherwise, we may further decrease $\zeta$ without violating constraints \eqref{equ: sslack5}. This indicates that $\zeta(\mathbf{p})$ is monotonic increasing with respect to $p_{\pi(m)}$. Thus, the optimal power $p_{\pi(m)}^*$ should always be the lower bound, which can be obtained from constraints \eqref{equ: qos1}, i.e.,
\begin{align}
	p_{\pi(m)}^*=\frac{\left( 2^{r^*}-1 \right)
		\left(\sum_{n=m+1}^{K} p_{\pi(n)}^* {h}_{\pi(n)} +\sigma^2\right)}{{h}_{\pi(m)} },\forall m.\label{equ: power3}
\end{align}
Note that the power allocations in \eqref{equ: power3} are mutually coupled with that of each device. Hence, in the following lemma, we will decouple this relationship, and give a closed-form of $p^*_{\pi(m)}$.
\begin{lemma}\label{equ: lemma31}
	Let $c_m\triangleq\frac{(2^{r^*}-1)\sigma^2}{\rho_0} 2^{(K-m)r^*}$ and then the closed-form solution of $p^*_{\pi(m)}$ is given by
	\begin{align}
		p^*_{\pi(m)} = c_m(H^2+||\mathbf{q}-w_{\pi(m)}||^2),\forall m.\label{equ: power4}
	\end{align}
\end{lemma}
\begin{proof}
	See Appendix \ref{appendix:lemma1}.
\end{proof}
With Lemma \ref{equ: lemma31}, problem \eqref{equ: problem35} reduces to the following problem for optimizing $\mathbf{q}$ only, i.e.,
\begin{subequations}\label{equ: problem36}
	\begin{align}
		&\mathop {\min }\limits_{\zeta,\mathbf{q}} \quad \zeta\\
		&\text{s.t.} \quad P_c+ c_m(H^2+||\mathbf{q}-w_{\pi(m)}||^2)\leq \zeta E_{\pi(m)},\forall m,\label{equ: dist2}\\
		&\quad\quad H^2+||\mathbf{q}-w_{\pi(m)}||^2 \leq \bar{P}_{\pi(m)}, \forall m,\label{equ: dist3}\\
		&\quad\quad \eqref{equ: serialdist2},
	\end{align}
\end{subequations}
where $\bar{P}_{\pi(m)}\triangleq \frac{1}{c_m}\min\{P_{\max}, \frac{I_{t h}}{z_{\pi(m)}-\varepsilon_{p}^{2} \ln \varrho}\} $.
Problem \eqref{equ: problem36} is convex and satisfying Slater's constraint qualification.
As a result, the optimal solution of problem \eqref{equ: problem36} can be obtained by solving its dual problem.
Denote $\bm{\lambda}$, $\bm{\mu}$, and $\bm v$ as the non-negative Lagrange multipliers for constraints \eqref{equ: dist2}, \eqref{equ: dist3}, and \eqref{equ: serialdist2}, respectively.
The Lagrange function of problem \eqref{equ: problem36} can be derived in \eqref{equ: lan} at the top of the next page.
Accordingly, the dual function is given by
\begin{align}
	f(\bm{\lambda},\bm{\mu},\bm v)=\min_{\zeta,\mathbf{q}} \mathcal{L}(\zeta,\mathbf{q},\bm{\lambda},\bm{\mu},\bm v),\label{equ: dualproblem}
\end{align}
for which the following lemma holds.
\begin{lemma}\label{lemma_dual}
	To make $f(\bm{\lambda},\bm{\mu},\bm v)$ bounded from the above, i.e., $f(\bm{\lambda},\bm{\mu},\bm v)<+\infty$, it holds that $\sum_{m=1}^{K} \lambda_{m} E_{\pi(m)}=1$.
\end{lemma}
\begin{proof}
	If $\sum_{m=1}^{K} \lambda_{m} E_{\pi(m)} > 1$ or $\sum_{m=1}^{K} \lambda_{m} E_{\pi(m)} < 1$, then we have $f(\bm{\lambda},\bm{\mu},\bm{v})\rightarrow +\infty $ by setting $\zeta \rightarrow+\infty$ or $\zeta \rightarrow-\infty$. This contradicts that the dual function is bound, which completes the proof.
	\hfill $\blacksquare$
\end{proof}
Problem \eqref{equ: dualproblem} is convex such that the optimal solution can be obtained by applying Karush-KuhnTucker (KKT) conditions, i.e.,
\begin{align}
	\mathbf{q}^*=
	\frac{\sum_{m=1}^K( c_m\lambda_{m}+\mu_m)\mathbf{w}_{\pi(m)}+\sum_{m=1}^{K-1}v_m (\mathbf{w}_{\pi(m)}-\mathbf{w}_{\pi(m+1)}) }{\sum_{m=1}^K( c_m\lambda_{m}+\mu_m)}.\label{equ: solution_y31}
\end{align}
After obtaining $\mathbf{q}^*$ for any given $\{\boldsymbol \lambda, \bm{\mu},\bm{v}\}$, the dual problem of \eqref{equ: problem36} is given by
\begin{subequations}\label{dualproblem8}
	\begin{align}
		\max _{\boldsymbol \lambda, \bm{\mu},\bm{v}}&\quad f(\bm{\lambda},\bm{\mu},\bm{v})\\
		\text{s.t.}&\quad \sum_{m=1}^{K} \lambda_{m}E_{\pi(m)}=1,\label{equ: lambda}\\
		&\quad \bm{\lambda} \succeq 0,\bm{\mu} \succeq 0, \bm{v} \succeq 0.
	\end{align}
\end{subequations}
Problem \eqref{dualproblem8} can be solved via a subgradient-based method such as the ellipsoid method. Specifically, the subgradient of the objective function is denoted by $\boldsymbol \eta_0= [ \Delta \bm{\lambda},\Delta \bm{\mu},\Delta \bm{v}]^T$ with $\Delta \lambda_{m} = -c_m||\mathbf{q}^* -\mathbf{w}_{\pi(m)}||^2-(c_mH^2+P_c)$, $ \Delta \mu_m =-||\mathbf{q}^* -\mathbf{w}_{\pi(m)}||^2-(H^2-\bar{P}_{\pi(m)})$, and $ \Delta \bm{v}_m = -2(\mathbf{w}_{\pi(m+1)}-\mathbf{w}_{\pi(m)})^T\mathbf{q}+||\mathbf{w}_{\pi(m+1)}||^2-||\mathbf{w}_{\pi(m)}||^2, \forall m$. 
Moreover, the equality constraint \eqref{equ: lambda} is equivalent to two inequality constraints, i.e., $\sum_{m=1}^{K} \lambda_{m}E_{\pi(m)} \geq1$ and $\sum_{m=1}^{K} \lambda_{m}E_{\pi(m)} \leq1$, whose subgradients are denoted by $\boldsymbol \eta_1 = [ \Delta \bm \lambda,\Delta \bm \mu, \Delta \bm{v}]^T$ with $\Delta \lambda_{m} = -E_{\pi(m)}, \Delta \mu_m=0, \Delta {v}_k = 0, \forall k$ and $\boldsymbol \eta_2 = [ \Delta \bm \lambda,\Delta \bm \mu,\Delta \bm{v}]^T$ with $\Delta \lambda_{m} = E_{\pi(m)}, \Delta \mu_m=0, \Delta {v}_m = 0,\forall m $, respectively.
After obtaining the optimal location $\mathbf{{q}}^*$, the corresponding optimal transmit power $\mathbf{{p}}^*$ can be directly calculated using \eqref{equ: power4} and objective value $\zeta^*$ can be obtained as $\zeta^* =\mathop {\max }_{m}\{\frac{p^*_{\pi(m)}+P_c}{E_{\pi(m)}}\}$.
The details of the procedures for solving problem \eqref{equ: problem34} are summarized in Algorithm \ref{algorithmsubgradient}. The complexity of Algorithm \ref{algorithmsubgradient} is $O(K^4)$.
\begin{figure*}
	\hrule
	\begin{small}
	\begin{align}\label{equ: lan}
		&\mathcal{L}(\zeta,\mathbf{q},\bm{\lambda},\bm{\mu})=(1-\sum_{m=1}^{K}\lambda_m E_{\pi(m)})\zeta+\sum_{m=1}^{K}(\lambda_m c_m+ \mu_m)||\mathbf{q}-\mathbf{w}_{\pi(m)}||^2+\sum_{m=1}^{K}\left( \lambda_m(c_mH^2+P_c)+\mu_m(H^2-\bar{P}_{\pi(m)})\right)\nonumber\\
		&+2\sum_{m=1}^{K-1}v_m (\mathbf{w}_{\pi(m+1)}-\mathbf{w}_{\pi(m)})^T\mathbf{q} +\sum_{m=1}^{K-1}v_m
		\left(||\mathbf{w}_{\pi(m)}||^2-||\mathbf{w}_{\pi(m+1)}||^2 \right).
	\end{align}
\end{small}
	\hrule
\end{figure*}

\begin{algorithm}[!t]
	\caption{Joint Location and Transmit Power Algorithm for Solving Problem \eqref{equ: problem34}}
	\label{algorithmsubgradient}
	\begin{algorithmic}[1]
		\STATE Initialize $\bm{\lambda},\bm{\mu},\bm v$ and the ellipsoid.		
		\REPEAT
		\STATE Obtain $\mathbf{q}^*$ by \eqref{equ: solution_y31}.
		\STATE Compute the subgradients of objective function and constraint functions in problem \eqref{dualproblem8}.
		\STATE Update $\{\bm{\lambda},\bm{\mu},\bm v\}$ by using the constrained ellipsoid method.
		\UNTIL $\bm{\lambda},\bm{\mu},\bm v$ converge within a prescribed accuracy.
		\STATE Set \{$\bm{\lambda}^*, \bm{\mu}^*, \bm v^*\} \leftarrow \{ \bm{\lambda}, \bm{\mu}, \bm v$\}.
		\STATE Obtain $\mathbf{q}^*$, $\mathbf{{p}}^*$ and $\zeta^*$ by using \eqref{equ: solution_y31} and \eqref{equ: power4}.
	\end{algorithmic}
\end{algorithm}

\begin{algorithm}[!t]
	\caption{Globally Optimal: Proposed Lagrange-duality-based Algorithm for Solving Problem \eqref{equ: problem33}}
	\label{algorithm1}
	\begin{algorithmic}[1]
		\STATE Set the number $l=0$.
		\STATE Find $K!$ decoding orders and represent them as the set $\mathcal{P}=\{\mathcal{P}_1,\cdots,\mathcal{P}_{K!}\}$.
		\FOR{$l=1$ to $K!$} 		
		\STATE Set $\{\pi^l(1),\pi^l(2),\cdots,\pi^l(K)\} \leftarrow \mathcal{P}_l$.
		\STATE Solve problem \eqref{equ: problem34} by applying Algorithm \ref{algorithmsubgradient} for given ${\bm{\pi}^l}$, and denote the optimal solution as $\{ \zeta^l,\mathbf{q}^l,\mathbf{p}^l\}$. 
		\IF{$\zeta^l \leq \zeta^{*}$}
		\STATE	Set $ \{\zeta^*,\mathbf{q}^*,\mathbf{p}^*,\bm{\pi}^*\} \leftarrow \{\zeta^l,\mathbf{q}^l,\mathbf{p}^l,\bm{\pi^l}\}$. 
		\ENDIF
		\ENDFOR
		\STATE 	Obtain the optimal solution as $\{\mathbf{q}^*, \mathbf{{p}}^*,\bm{\pi}^*\}$.
	\end{algorithmic}
\end{algorithm}

\subsection{Decoding Order Optimization and Overall Algorithm }\label{sec: decoding}
We will exhaustively search all the $K!$ possible decoding orders while obtaining the globally optimal UAV location $\mathbf{q}$ and transmit power $\mathbf{p}$ via Algorithm 1 in each inner loop.
The globally optimal value and solution of problem \eqref{equ: problem33} are the maximum value determined by all possible decoding orders and the corresponding $\{ \mathbf q, \mathbf p, \bm{\pi}\}$, respectively.
The details of the procedures are summarized in Algorithm \ref{algorithm1}.
The complexity of Algorithm \ref{algorithm1} is $O(K^4K!)$.

\begin{remark}
	Although Algorithm \ref{algorithm1} entails high complexity, it provides a practically feasible method to obtain the global optimal solution.
	Moreover, Algorithm \ref{algorithm1} is feasible for relatively small-scale IoT networks and can also serve as a benchmark for evaluating other algorithms.
\end{remark}

\section{Low-Complexity Suboptimal Solution}\label{sec: sub-optimal}
The Lagrange-duality-based algorithm in the preceding section provides an optimal solution to problem \eqref{equ: problem33}.
However, the computational complexity of the algorithm is relatively high.
In this section, as an alternative, we propose a low-complexity iterative algorithm for solving problem \eqref{equ: problem33}, which yields a sub-optimal solution for large-scale IoT networks.
Specifically, we first reformulate problem \eqref{equ: problem33} into an equivalent but more tractable form by introducing binary variables to replace the permutation variables and showing their explicit relationship.
 Next, we derive a closed-form solution of the optimal power and then apply the SCA technique and penalty function method to obtain a sub-optimal solution.

\subsection{ Equivalent Reformulation of Problem \eqref{equ: problem33}}
As shown in Section \ref{sec: optimal}, the main complexity of Algorithm 2 lies in the exhaustive search over $K!$ possible solutions.
To reduce the complexity, in this section, we first introduce auxiliary variables $\bm{\alpha}=\{\alpha_{k,j}, k,j \in \mathcal{K}\}$ and the following problem.
\begin{subequations}\label{equ: newproblem}
	\begin{align}
		&\mathop {\min }\limits_{\zeta,\mathbf{q},\mathbf{p},\bm{\alpha}}\quad \zeta\\
		&\text{s.t.}\quad p_{k}+P_c \leq \zeta E_{k},\forall k,\label{equ: newsslack}\\
		&\qquad 0 \leq p_{k} \leq \tilde{P}_k ,\forall k,\label{equ: newpower}\\
		&\qquad \log_2\left(1+\frac{p_{k}h_{k}}{\sum_{j=1,j\neq k}^K \alpha_{k,j}p_jh_{j}+\sigma^2} \right) \geq r^*, \forall k,\label{equ: qosnew}\\
		&\qquad \alpha_{k,j}=\left\{\begin{array}{ll}{0,} & {\text { if } d_{k} > d_{j},} \\ {1,} & {\text { if } d_{k} < d_{j},}\\ {0 \text{ or } 1,} & {\text { if } d_{k} = d_{j},} \end{array}\right. \quad \forall k \neq j,\label{equ: SIC}\\
		&\qquad \alpha_{k,k} = 0, \forall k,\label{equ: SICself}\\
		&\qquad \alpha_{k,j}+\alpha_{j,k} = 1,\forall k\neq j,\label{equ: SICpair}\\
		&\qquad \alpha_{k,j}+\alpha_{j,i}-1 \leq \alpha_{k,i}, \forall k,j,i,\label{equ: SICnew}
	\end{align}
\end{subequations}
where $k,j,i\in \mathcal{K}$ denote the device indexes and are fixed.
\begin{remark}
	In fact, the SIC procedure can be guaranteed by satisfying constraints \eqref{equ: SIC}-\eqref{equ: SICnew}, as elaborated below.
	First, let $\alpha_{k,j} = 1$ denote that the signal of device $k$ is decoded before that of device $j$; otherwise, $\alpha_{k,j} =
	0$. Then, the practical meanings of constraints \eqref{equ: SIC}-\eqref{equ: SICpair} can refer to work \cite{tang2020cognitive}.
	Constraints \eqref{equ: SICnew} ensure that when $\alpha_{k,j}=1$ and $\alpha_{j,i}=1$, it must have $\alpha_{k,i}=1$; otherwise, $\alpha_{k,i}\geq 0$\footnote{Such a relationship can also written as $\alpha_{k,j}\alpha_{j,i} \leq \alpha_{k,i},\forall k,j,i$ \cite{nguyen2018novel}. However, the constraints in \cite{nguyen2018novel} are non-convex. By contrast, our proposed constraints \eqref{equ: SICnew} are linear and easier to handle. }.
	This implies that
	if the signal of device $k$ is decoded before that of device $j$, and the signal of device $j$ is decoded before that of device $i$, then the signal of device $k$ must be decoded before that of device $i$. 
\end{remark}

In the following, we first present the explicit relationship between \bm{\pi} and $\bm \alpha$ and then prove that problems \eqref{equ: problem33} and \eqref{equ: newproblem} share the same optimal value.
To proceed, an important and useful transformation of constraint \eqref{equ: qosnew} is given in the following proposition.
\begin{proposition} \label{lem:lemma3}
	If $\bm \alpha$ satisfies constraints \eqref{equ: qosnew}-\eqref{equ: SICnew}, then constraints \eqref{equ: qosnew} are equivalent to 
	$f(k)= K-\sum_{j=1}^{K}{\alpha}_{k,j}$ and 
	\begin{align}
		\log_2\left(1+\frac{ p_{k}h_{k}}{\sum_{f(j)= f(k)+1}^K {p}_{j} {h}_{j}+\sigma^2} \right) \geq r^*, \forall k.\label{equ: temp11}
	\end{align}
\end{proposition}
\begin{proof}
	See Appendix \ref{appendix:lemma2}.
\end{proof}

Then, we come to reveal how to obtain a feasible solution of problem \eqref{equ: problem33}/\eqref{equ: newproblem} from the solution of problem \eqref{equ: newproblem}/\eqref{equ: problem33} in the following propositions.
\begin{proposition}\label{lem:lemma4}
	For any feasible solution $\{\bm \alpha, \mathbf p, \mathbf{q}\}$ of problem \eqref{equ: newproblem}, we can construct $\{\bm{\pi}', \mathbf{p}',\mathbf{q}\}$ with $\mathbf{p}'=[p'_1, \cdots, p'_K]^T$, as a feasible solution of problem \eqref{equ: problem33} by setting $\pi'=f^{-1}$ and $p'_{\pi'(m)}=p_{f^{-1}(m)}$. Moreover, the objective value of problem \eqref{equ: problem33} obtained at $\{\bm{\pi}',\mathbf{p}',\mathbf q\}$ is the same as that of problem \eqref{equ: newproblem} obtained at $\{\bm \alpha, \mathbf p, \mathbf{q}\}$\footnote{ Furthermore, owing to the definition of \bm{\pi} and $\pi=f^{-1}$ in Propositions \ref{lem:lemma4}, we have $m=\pi^{-1}(k)=f(k)$, which means	that the signal from device $k$ should be decoded at the $f(k)$-th (i.e., $K-\sum_{j=1}^{K}{\alpha}_{k,j}$). 
	}.
\end{proposition}
\begin{proof}
	See Appendix \ref{appendix:lemma3}.
\end{proof}
\begin{proposition}\label{lem:lemma5}
	For any feasible solution $\{\bm{\pi},\mathbf{p},\mathbf{q}\}$ of problem \eqref{equ: problem33}, 
	we can construct $\{\tilde {\bm {\alpha}},\tilde {\mathbf p},\mathbf{{q}}\}$ with $\tilde {\bm {\alpha}}=\{\tilde {\alpha}_{k,j}, k,j \in \mathcal{K}\}$ and $\tilde {\mathbf p}=[\tilde p_1, \cdots, \tilde p_K]^T$ as a feasible solution of problem \eqref{equ: newproblem} by setting $\tilde p_k=p_{\pi(m)}$ and
	\begin{align}
		\tilde{\alpha}_{\pi(m),\pi(n)} = \left\{\begin{array}{ll}{1,} & {\text { if } m<n,} \\ {0,} & {otherwise. } \\ \end{array}\right.
	\end{align}
	Moreover, 
	the objective value of problem \eqref{equ: newproblem} obtained at $\{\tilde {\bm {\alpha}},\tilde {\mathbf p},\mathbf{{q}}\}$ is identical to that of problem \eqref{equ: problem33} obtained at $\{\bm{\pi},\mathbf{p},\mathbf{q}\}$.
\end{proposition}
\begin{proof}
	See Appendix \ref{appendix:lemma5}.
\end{proof}

We are now ready to clarify the equivalence between problems \eqref{equ: problem33} and \eqref{equ: newproblem}.

\begin{proposition}\label{lem:lemma6}
	Problem \eqref{equ: newproblem} can achieve the same optimal value as problem \eqref{equ: problem33}.
\end{proposition}
\begin{proof}
	Define $L_1$ and $L_2$ as the optimal value of problem \eqref{equ: problem33} and problem \eqref{equ: newproblem}, respectively.
	Based on Proposition \ref{lem:lemma4} and \ref{lem:lemma5}, we have $L_1 \leq L_2$ and $L_1 \geq L_2$, respectively. Thus, $L_1 = L_2$.			\hfill $\blacksquare$
\end{proof}


It should be noted that in the literature, both binary variables $\bm \alpha$ as well as permutation variables \bm{\pi} are adopted for decoding order formulations, such as \cite{8848428,8685130,lu2020uav,8918266,tang2020cognitive,xu2020joint,zhang2019optimal,nguyen2018novel}.
However, these two different problem formulations are used alternatively without any explicit relationship and thus have always been treated as two independent programs. 
Furthermore, only a sub-optimal solution of $\bm{\alpha}$-related programs can be obtained by relaxing the binary constraints and applying the SCA technique in general. It yet remains unknown whether these two problem formulations can achieve the same optimal performance.
In this paper, we prove that they can be used interchangeably via Propositions \ref{lem:lemma4}-\ref{lem:lemma6} and thus unify these two problems for the first time in the literature.

\subsection{Closed-form Solution of Transmit Power}
Based on proposition \ref{lem:lemma6}, we only need to focus on solving problem \eqref{equ: newproblem} with the non-convex constraints \eqref{equ: qosnew}-\eqref{equ: SIC} and binary variables $\bm{\alpha}$.
In general, there is no standard method to solve such a problem optimally.
In this subsection, we will derive a explicit analytical solution of $p_k$ as a function of $\bm \alpha$ and $\mathbf{{q}}$ based on \eqref{equ: temp11}, by which problem \eqref{equ: newproblem} can be greatly simplified.

By Proposition \ref{lem:lemma3}, problem \eqref{equ: newproblem} can be equivalently written as
\begin{subequations}\label{equ: problem38}
	\begin{align}
		&\mathop {\min }\limits_{\zeta,\mathbf{q},\mathbf{p},\bm{\alpha},f}\quad \zeta\\
		&\text{s.t.}
		\quad p_k\geq \frac{\left( 2^{r^*}-1 \right)\left(\sum_{ f(j)= f(k)+1}^K {p}_{j} {h}_{j}+\sigma^2\right)}{{h}_{k} },\forall k,\label{equ: qosnew2} \\
		&\qquad f(k)= K-\sum_{j=1}^{K}{\alpha}_{k,j},\forall k,\label{equ: temp12}\\
		&\qquad \eqref{equ: newsslack}, \eqref{equ: newpower}, \eqref{equ: SIC}-\eqref{equ: SICnew}. 
	\end{align}
\end{subequations}
As we can see, the objective function $\zeta(\mathbf{p}) =\mathop {\max }_{k}\{\frac{p_{k}+P_c}{E_{k}}\}$ of problem \eqref{equ: problem38} is
increasing with $p_k$.
Thus, the optimal $\mathbf p^*$ should always be the lower bound, which can be obtained when constraints \eqref{equ: qosnew2} satisfy strict equality.
To obtain $\mathbf p^*$, define $p_{f(k)}^{s}=p_k$ and $h_{f(k)}^s(\mathbf{w}_{f(k)}^s)=h_k(\mathbf{w}_{k})$.
Then constraint \eqref{equ: qosnew2} can be rewritten as 
\begin{align}
	p_{f(k)}^{s} \geq \frac{( 2^{r^*}-1 )}{h^s_{f(k)} }\left(\sum_{ f(j)= f(k)+1}^K p_{f(j)}^{s} h_{f(j)}^s+\sigma^2\right).
\end{align}
By following the similar lines as previous proofs of Lemma 1, we can obtain the optimal $p_{f(k)}^{s*}=c_{f(k)}(H^2+||\mathbf{q}-\mathbf{w}_{f(k)}^s||^2)$, where $c_{f(k)}\triangleq\frac{(2^{r^*}-1)\sigma^2}{\rho_0}2^{(K-f(k))r^*}$.
This yields
\begin{align}
	p_{k}^*=c_{f(k)}(H^2+||\mathbf{q}-\mathbf{w}_{k}||^2).\label{equ: optimalbeta}
\end{align}	
By adopting a similar approach in our previous work \cite{tang2020cognitive}, the non-smooth constraints \eqref{equ: SIC} can be equivalently transformed into a smooth form as follows.
\begin{subequations}\label{equ: tmp3}
	\begin{align}
		& \sum\limits_{k,j}^K \left(\alpha_{k,j}-\alpha_{k,j}^2 \right) \leq 0,\label{equ: timescheduletemp}\\
		&\sum_{k,j}^K g_{k,j}(\alpha_{k,j},\mathbf{q}) \leq 0,\label{equ: functiontemp}\\
		&0 \leq \alpha_{k,j} \leq 1, \forall k, j,\label{equ: binarytmp}
	\end{align}
\end{subequations}
where $g_{k,j}(\alpha_{k,j},\mathbf{q})\triangleq \theta_{k,j}(2\alpha_{k,j}-1)+\left|\theta_{k,j}\right|$  and $\theta_{k,j} \triangleq 2(\mathbf{w}_{j}-\mathbf{w}_{k})^T\mathbf{q}+||\mathbf{w}_{k}||^2-||\mathbf{w}_{j}||^2 $.
Note that $\theta_{k,j}$ is linear with respect to $\mathbf{q}$ and $|\theta_{k,j}|$ is convex with respect to $\mathbf{q}$ \cite{tang2020cognitive}.
With \eqref{equ: optimalbeta} and \eqref{equ: tmp3}, problem \eqref{equ: problem38} reduces to
\begin{subequations}\label{equ: problem3110}
	\begin{align}
		\mathop {\min }\limits_{{\zeta,\mathbf{q},\bm{\alpha}},f}&\quad \zeta\\
		\text{s.t.} &\quad c_{f(k)}\left( H^2+||\mathbf{q}-\mathbf{w}_k||^2\right)+P_c \leq \zeta E_{k},\forall k,\label{equ: sslack8}\\
		&\quad c_{f(k)}\left( H^2+||\mathbf{q}-\mathbf{w}_k||^2\right) \leq \tilde{P}_k,\forall k,\label{equ: power8} \\
		&\quad  \eqref{equ: SICself}-\eqref{equ: SICnew}, \eqref{equ: temp12},\eqref{equ: tmp3}.
	\end{align}
\end{subequations}
\subsection{SCA-Based Optimization}

The main challenge of problem \eqref{equ: problem3110} lies in that $f$ is involved as the index in constraints \eqref{equ: sslack8} and \eqref{equ: power8}.
A close observation of problem \eqref{equ: problem3110} shows that we can solve the following sub-problems alternately: sub-problem 1 optimizes
$f$ with given $\bm \alpha$, whose optimal solution can be obtained from constraints \eqref{equ: temp12}, i.e.,
\begin{align}
	f^*(k)= K-\sum_{j=1}^{K}{\alpha}_{k,j},\label{equ: temp13}
\end{align}
and subproblem 2 optimizes $\{\mathbf{q},\bm{\alpha}\}$ with given $f$, i.e.,
\begin{subequations}\label{equ: problem311}
	\begin{align}
		\mathop {\min }\limits_{{\zeta,\mathbf{q},\bm{\alpha}}}&\quad \zeta\\
		\text{s.t.}
		&\quad  \eqref{equ: SICself}-\eqref{equ: SICnew}, \eqref{equ: tmp3},\eqref{equ: sslack8}, \eqref{equ: power8}.
	\end{align}
\end{subequations}
The remaining task is to solve problem \eqref{equ: problem311} with non-convex constraints \eqref{equ: timescheduletemp} and \eqref{equ: functiontemp}.
With loss of generality, we can leverage the SCA technique to approximate the non-convex term to a convex form in each iteration and then iteratively solve a series of approximated convex problems.
However, the SCA technique cannot be straightforwardly applied to constraints \eqref{equ: timescheduletemp} and \eqref{equ: functiontemp} since there may be some iterations where the approximated problem is infeasible due to \eqref{equ: binarytmp}.
To overcome this obstacle, we introduce slack variables $\phi$ and $\varphi$ and take into account the relaxed version of problem \eqref{equ: problem311} as
\begin{subequations}\label{equ: problem312}
	\begin{align}
		\mathop {\min }\limits_{{\zeta,\mathbf{q},\bm{\alpha}}} &\quad \zeta+ \rho_1 \phi +\rho_2 \varphi \\
		\text{s.t.}	
		& \quad \eqref{equ: SICself}-\eqref{equ: SICnew}, \eqref{equ: binarytmp}, \eqref{equ: sslack8}-\eqref{equ: power8},
	\end{align}
\end{subequations}
where $\phi \triangleq\sum\limits_{k,j}^{K} (\alpha_{k,j} -\alpha_{k,j}^2)$ and $\varphi \triangleq\sum\limits_{k,j}^{K}g_{k,j}(\alpha_{k,j},\mathbf{q})$, and
$\rho_1 > 0$ and $\rho_2>0$ are penalty parameters.
It is worth noting that it must have $\phi \geq 0$ and $\varphi \geq 0$ due to $\alpha_{k,j} -\alpha_{k,j}^2 \geq 0 $ and $g_{k,j}(\alpha_{k,j},\mathbf{q}) \geq 0$ within the feasible region of problem \eqref{equ: problem312}.
Inspired by \cite{8918266}, it can be proved that problem \eqref{equ: problem312} is equivalent to \eqref{equ: problem311}, when $\rho_1 \geq \rho_1^*$ and $\rho_2\geq \rho_2^*$, with $\rho_1^*$ and $\rho_2^*$ denoting the optimal Lagrange multiplier of constraints \eqref{equ: timescheduletemp} and \eqref{equ: functiontemp}, respectively.
In the following, we transform the non-convex constraints \eqref{equ: timescheduletemp} and \eqref{equ: functiontemp} into convex constraints by deriving the global lower bounds at a given point.
Specifically, based on the fact that the first-order Taylor expansion of concave function is its global over-estimator, we have the following upper bound for $\phi$ at any local point $\{\bar{\alpha}_{k,j}\}$, i.e.,
\begin{align}
	\phi \leq \sum_{k,j}^K \left( \alpha_{k,j}^2+\bar{\alpha}_{k,j}^{2}-2\bar{\alpha}_{k,j}\alpha_{k,j} \right) \triangleq \bar{\phi}.\label{equ: SICordertemp}
\end{align}
In addition,  we can obtain the upper bound of $\alpha_{kj}\theta_{k,j}$ at any local point $\{\bar{\mathbf{q}},\bar \alpha_{k,j}\}$ as
\begin{align}
	&2\alpha_{kj}\theta_{k,j}=\frac{1}{2}\left[ (\theta_{k,j}+\alpha_{kj})^2-(\theta_{k,j}-\alpha_{kj})^2\right]\nonumber\\
	&\leq \frac{1}{2}\left[ (\theta_{k,j}+\alpha_{kj})^2+({\bar \theta_{k,j}}-{\bar \alpha_{k,j}})^2\right] -({\bar \theta_{k,j}}-{\bar \alpha_{k,j}})(\theta_{k,j}-\alpha_{kj}) \triangleq D_{kj},\label{equ: betaslack}
\end{align}
where ${\bar \theta_{k,j}} \triangleq2(\mathbf{w}_{j}-\mathbf{w}_{k})^T\mathbf{\bar q}+||\mathbf{w}_{k}||^2-||\mathbf{w}_{j}||^2 $.

With the upper bounds in \eqref{equ: SICordertemp} and \eqref{equ: betaslack}, as well as any given local point $\{\bar{\mathbf{q}},\bar \alpha_{k,j}\}$, problem \eqref{equ: problem312} can be approximated as the following problem.
\begin{subequations}\label{equ: problem313}
	\begin{align}
		\mathop {\min }\limits_{{\zeta, \mathbf{q},\bm{\alpha}}} &\quad \zeta+ \rho_1 \bar \phi +\rho_2 \bar \varphi \\
		\text{s.t.}& \quad \eqref{equ: SICself}-\eqref{equ: SICnew}, \eqref{equ: binarytmp}, \eqref{equ: sslack8}-\eqref{equ: power8},
	\end{align}
\end{subequations}
where $\bar \varphi = \sum\limits_{k,j}^{K}\left(D_{k,j}-\theta_{k,j}+\left|\theta_{k,j}\right| \right) $.
Note that problem \eqref{equ: problem313} is a convex optimization problem, which can be effectively solved by standard convex optimization tools such as CVX.
It is worth noting that
the feasible set of problem \eqref{equ: problem313} is a subset of that of problem \eqref{equ: problem312}.
Therefore, the objective value of problem \eqref{equ: problem313} gives a lower bound to that of problem \eqref{equ: problem312}.

\subsection{Overall Algorithm}\label{sec: algorithm2}
The proposed iterative algorithm for problem \eqref{equ: problem33}/\eqref{equ: newproblem} is concluded in Algorithm \ref{algorithm3}.
Specifically, problem \eqref{equ: newproblem} is solved by optimizing $f$ and subproblem \eqref{equ: problem313} iteratively in an alternate manner.
Since CVX invokes an interior-point method to solve the optimization problem, the computational complexity of the proposed algorithm is $O(N_{ite}K^{7})$, where $N_{ite}$ is the number of iterations for convergence in Algorithm \ref{algorithm3}.

\begin{algorithm}[!t]
	\caption{Sub-optimal: Proposed Iterative Algorithm for Solving Problem \eqref{equ: problem33}.}
	\label{algorithm3}
	\begin{algorithmic}[1]
		\STATE Initialize the UAV location $\mathbf{\bar q}$ and $\bm{\bar \alpha}$. 
		\REPEAT
		\STATE	Initialize $l=1$, $\rho_1^0$ and $\rho_2^0$. Set $b_1>1$, $b_2>1$, $\rho_{1}^{\max}$ and $\rho_{2}^{\max}$.
		\STATE Update $ f$ based on \eqref{equ: temp13} for given $\bm{\bar \alpha}$.
		\REPEAT		
		
		\STATE Solve problem \eqref{equ: problem313} for given $\{\mathbf{\bar q}, \bm{\bar \alpha}\} $, and denote the optimal solution as $\{\mathbf{q}^*,\bm{\alpha}^*\}$.
		\STATE Set $\{{{\mathbf{\bar q},\bm{\bar \alpha}}}\} \leftarrow \{{{\mathbf{q}^*,\bm{\alpha}^*}}\} $.
		\STATE Update $l \leftarrow l+1$, $ \rho_1^{l+1}\leftarrow \min\{b_1 \rho_1^{l},\rho_{1}^{\max}\}$ and $ \rho_2^{l+1}\leftarrow \min\{b_2 \rho_2^{l},\rho_{2}^{\max}\}$.
		\UNTIL Convergence.
		\UNTIL The fractional decrease of the objective value is below a threshold.
	\end{algorithmic}
\end{algorithm}

\begin{remark}
	In Algorithm 3, the initial UAV location $\mathbf{q}^0$ and $\bm{\alpha}^0$ should be carefully set. 
	To this end,
	we propose a feasibility checking method for initial points. 	
	Specifically, $\{\mathbf{q}^0,\bm{\alpha}^0\}$ is feasible to problem \eqref{equ: newproblem} only when the low-bound of $p_k$ is no larger than its upper bound, i.e.,
	\begin{align}
		\frac{\rho_0(2^{r^*}-1)2^{r^*(K-f^0(k))}}{\sigma^2}\left( H^2+||\mathbf{q}^0-\mathbf{w}_{f^0(k)}||^2\right) \leq \tilde{P}_{f^0(k)}, 
	\end{align}
\end{remark}
where $f^0$ is obtained based on $\bm{\alpha}^0$.

\section{Numerical Results}\label{sec: result}
In this section, numerical results are provided to evaluate the performance of the proposed algorithms.
Unless being mentioned elsewhere, the reference SNR is set as $\gamma_0=60$ dB. For IoT devices, the maximum transmit power is $P_{\max}=1$ W, circuit power consumption  is $P_c = 0.9$ W, and  battery energy is $E=4 \times 10^3$ J.
Moreover, the block error probability is set as $\varepsilon_{p}^{2} = 10^{-2}$, the interference threshold is $I_{th} = 28$ dBm, and 
the	 probability interference is $\varrho = 0.1\%$.
In this network, we consider that $K=6$ IoT devices are randomly
distributed in a horizontal plane, marked by `$\blacksquare$'s. 

In the following, the proposed optimal algorithm \ref{algorithm1} (denoted as `Op-NOMA-J') and sub-optimal algorithm \ref{algorithm3} (denote as `Sub-NOMA-J') are compared with two benchmark algorithms\footnote{The global optimal solution of the two benchmark algorithms can be obtained by using a similar method in Section \ref{sec: lifetimeproblem}.}:
\begin{itemize}
	\item NOMA-P: For a UAV-enabled IoT network with cognitive NOMA transmission, the UAV location is fixed at the geometric center of IoT devices while the transmit
	power is optimized.
	\item FDMA: For a UAV-enabled IoT network with cognitive FDMA transmission \cite{8886053}, the UAV location and transmit power are jointly optimized.
\end{itemize}

\begin{figure}[!t]
	\centering  
	\subfigcapskip=+5pt 
	\subfigure[ Inner loop]{   
		\begin{minipage}{7cm}
			\centering    
			\includegraphics[scale=0.5]{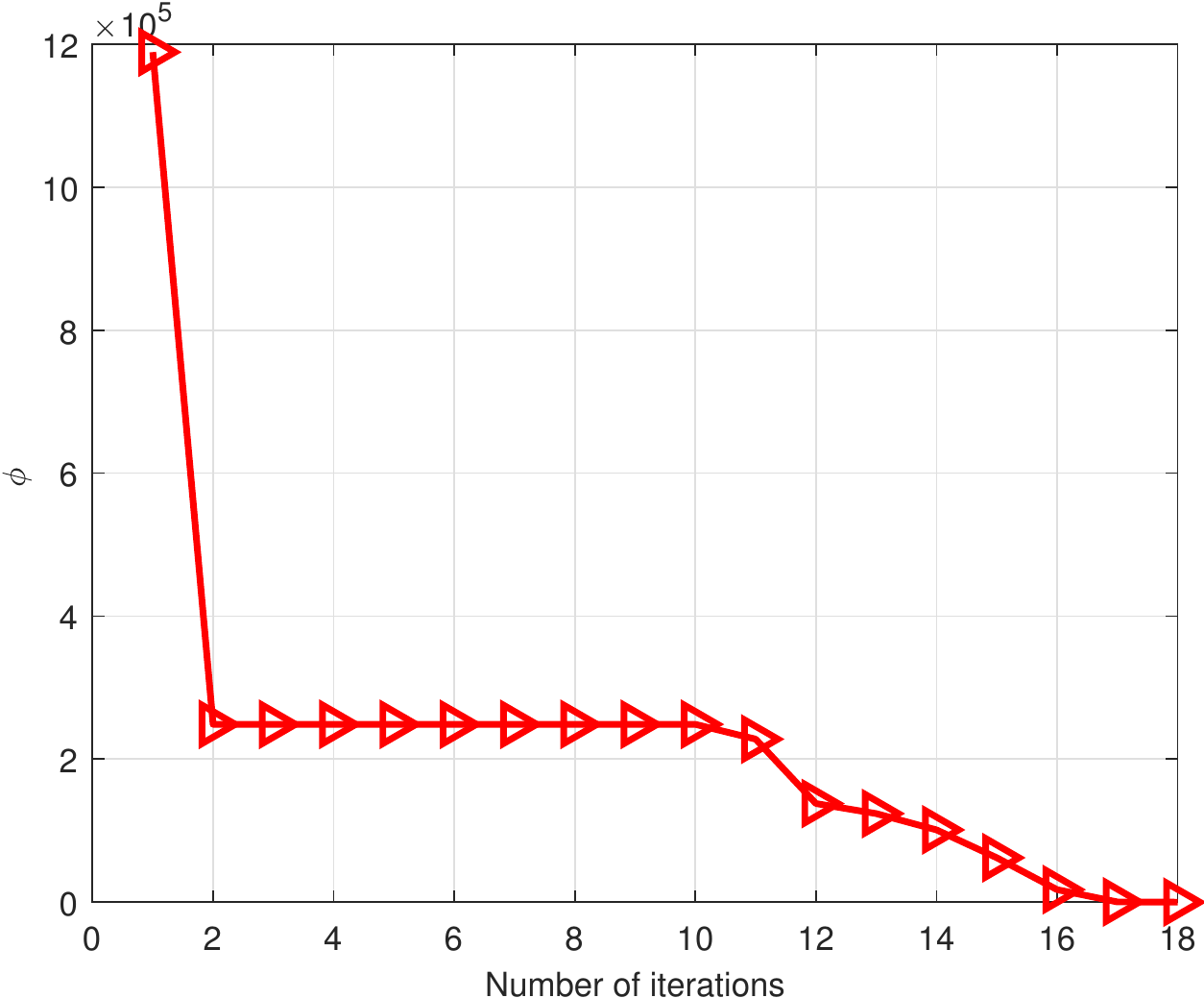}  
					\label{fig: convergence_penalty}
		\end{minipage}
	}
	\subfigure[ Outer loop]{ 
		\begin{minipage}{7cm}
			\centering    
			\includegraphics[scale=0.5]{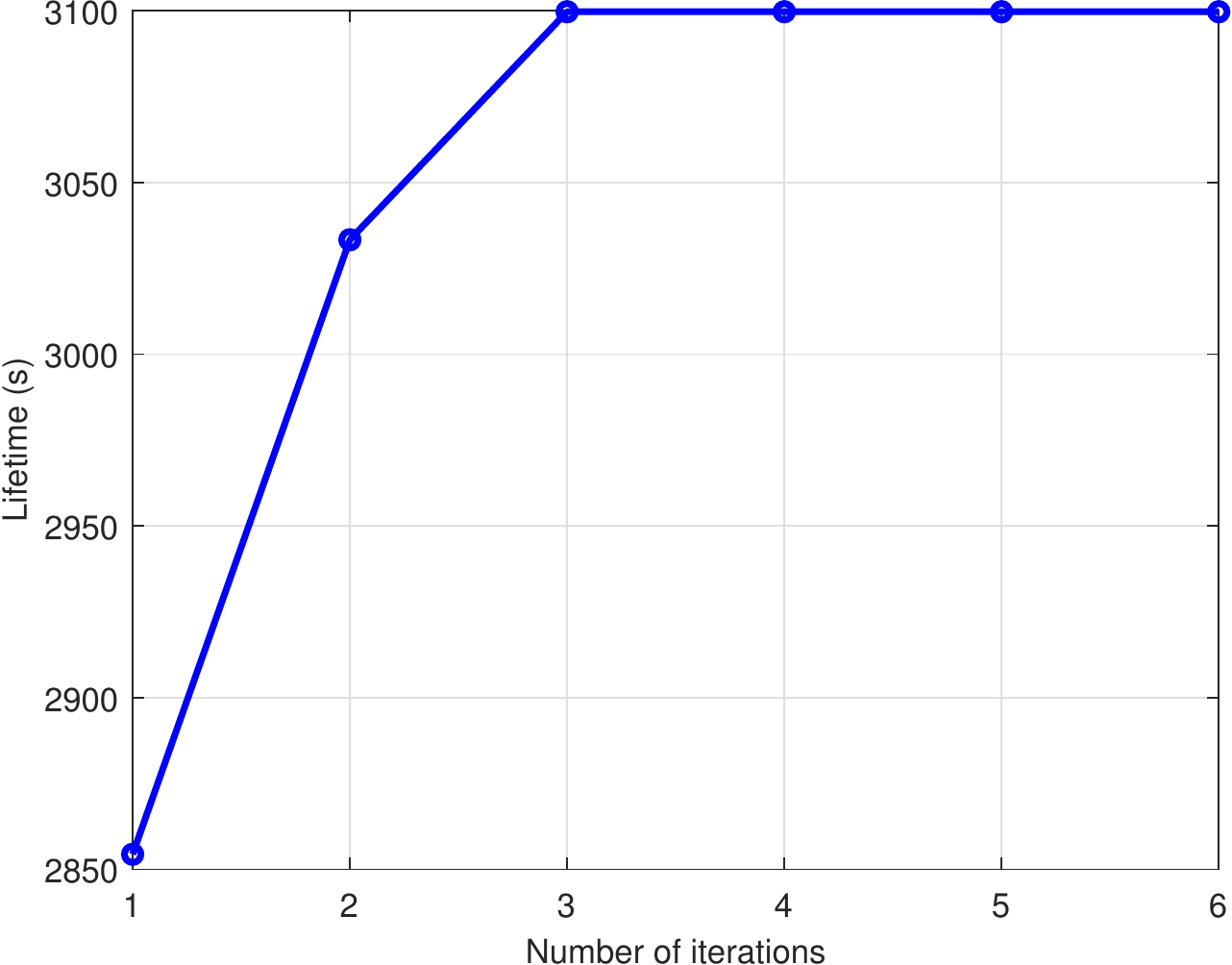}
					\label{fig: convergence}
		\end{minipage}
	}
	\caption{Convergence performance of Algorithm \ref{algorithm3}.}    
	\label{fig: convergenceall}    
\end{figure}


In Fig. \ref{fig: convergenceall}, we demonstrate the convergence performance of our proposed iterative algorithm \ref{algorithm3} in terms of the outer-loop iteration and inner-loop iteration with $r^* = 0.6 $ bps/Hz.
Fig. \ref{fig: convergence_penalty} shows the penalty value $\varphi$ versus the number of iterations in the first inner loop.
It can be observed that $\varphi \rightarrow 0$ in the end.
Fig. \ref{fig: convergence} shows the max-min lifetime versus the number of iterations in the outer loop.
From the figure, the lifetime increases quickly with the number of iterations and the algorithm converges within 6 iterations.


\begin{figure}[!t]
	\centering  
	\subfigcapskip=+5pt 
	\subfigure[ UAV location]{   
		\begin{minipage}{7cm}
			\centering    
			\includegraphics[scale=0.5]{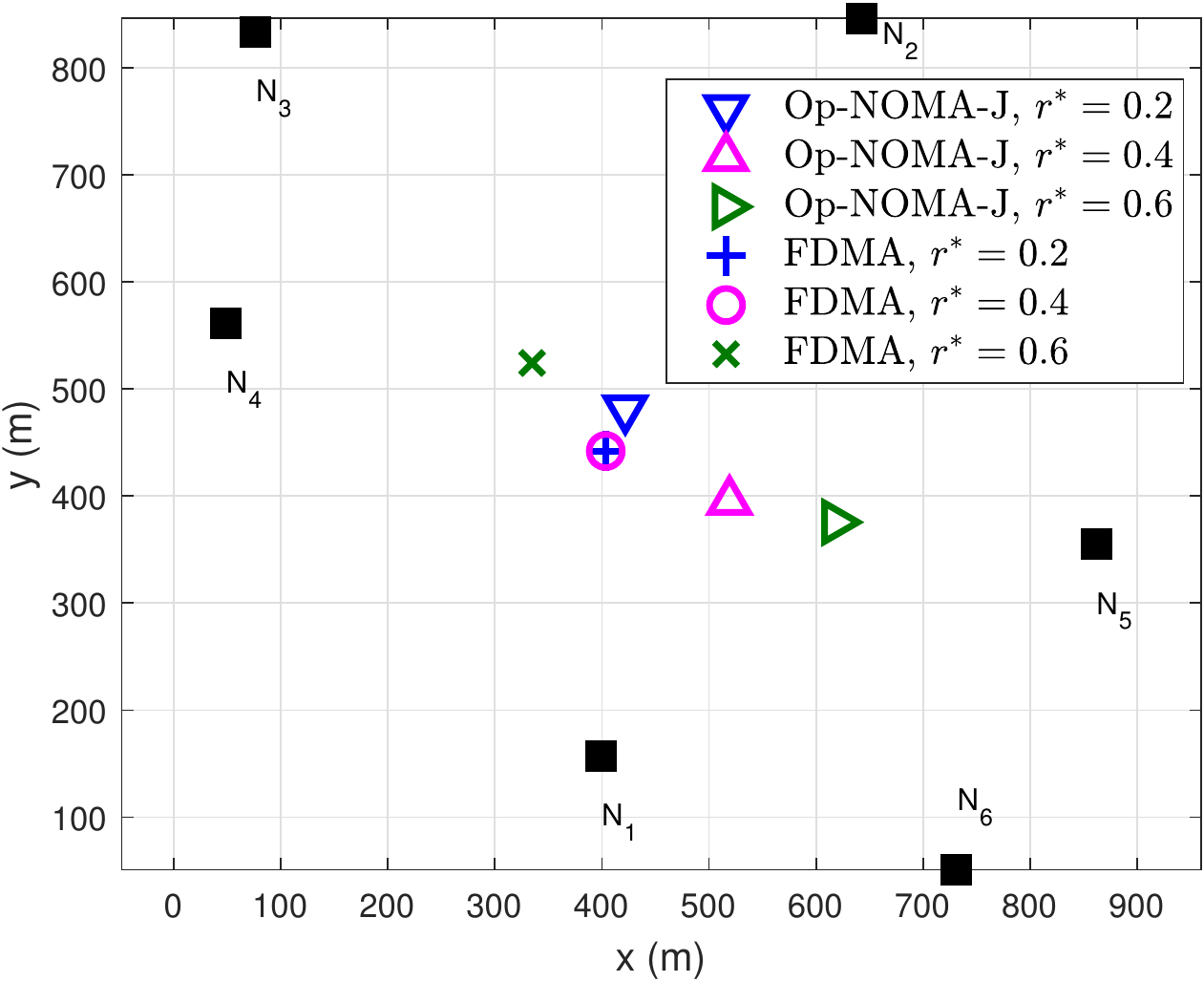}  
			\label{fig: Compare_location}
		\end{minipage}
	}
	\subfigure[ Power for each device]{ 
		\begin{minipage}{7cm}
			\centering    
			\includegraphics[scale=0.5]{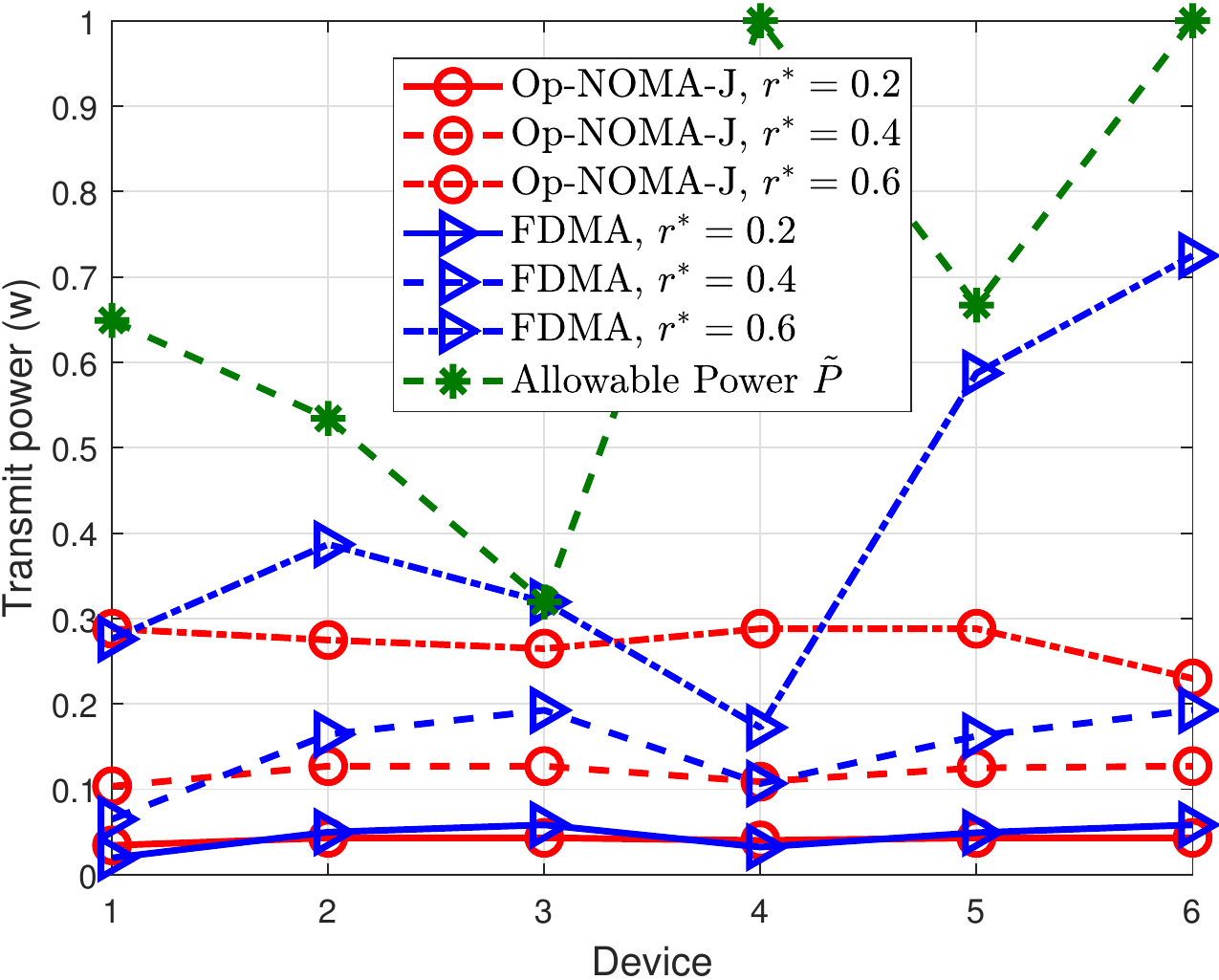}
			\label{fig: power}
		\end{minipage}
	}
	\caption{Optimized UAV location and power under different $r^*$ for Op-NOMA-J and FDMA.}    
	\label{fig: 1}    
\end{figure}


Fig. \ref{fig: 1} shows the UAV location and transmit power for Op-NOMA-J and FDMA under different QoS requirements $r^*$.
For FDMA, we observe from Fig. \ref{fig: Compare_location} that as $r^*$ increases, the UAV flies toward device $3$.  This can be explained as follows. With the increase in $r^*$, each device rises its transmit power to satisfy the QoS requirement.
Since the allowable transmit power of device 3, i.e.,  $\tilde{P}_3$, is the bottleneck of the network (as shown in Fig. \ref{fig: power}), then once the transmit power of device 3 $p^{F}_3$  approaches $\tilde P_3$, i.e., $p^{F}_3 = \tilde{P}_3$,  the UAV can only move closer to device $3$ to further improve the channel quality for the larger QoS requirement. On the contrary, for Op-NOMA-J,  the UAV moves from the vicinity of the geometric center to away from device $3$ and $4$ as $r^*$ increases. 
Note that the achievable rate of device $3$ is given by $R_3^N= \log_2(1+\frac{p^{N}_3h^{N}_3}{I_0+\sigma^2})$, where $I_0$ is the interference that device $3$ has experienced, $p^{N}_3$ is the transmit power and  $h^N_3$ is the channel gain. Owing to the small or even zero $I_0$, we can avoid reaching the upper bound of $p^{N}_3$ too early when $h^N_3$ reduces.
Besides, this result reveals that with the growth of $r^*$, the distance discrepancy between the strongest and weakest devices to the UAV increases, or the discrepancy between the strongest and weakest channel gain increases.
Moreover, from Fig. \ref{fig: power}, when $r^*$ is large, the transmit powers of all devices are quite distinct in Op-NOMA-J and FDMA, which is due to their different resource allocation mechanisms.
	
\begin{figure}[!t]
	\centering
	\includegraphics[width=3in]{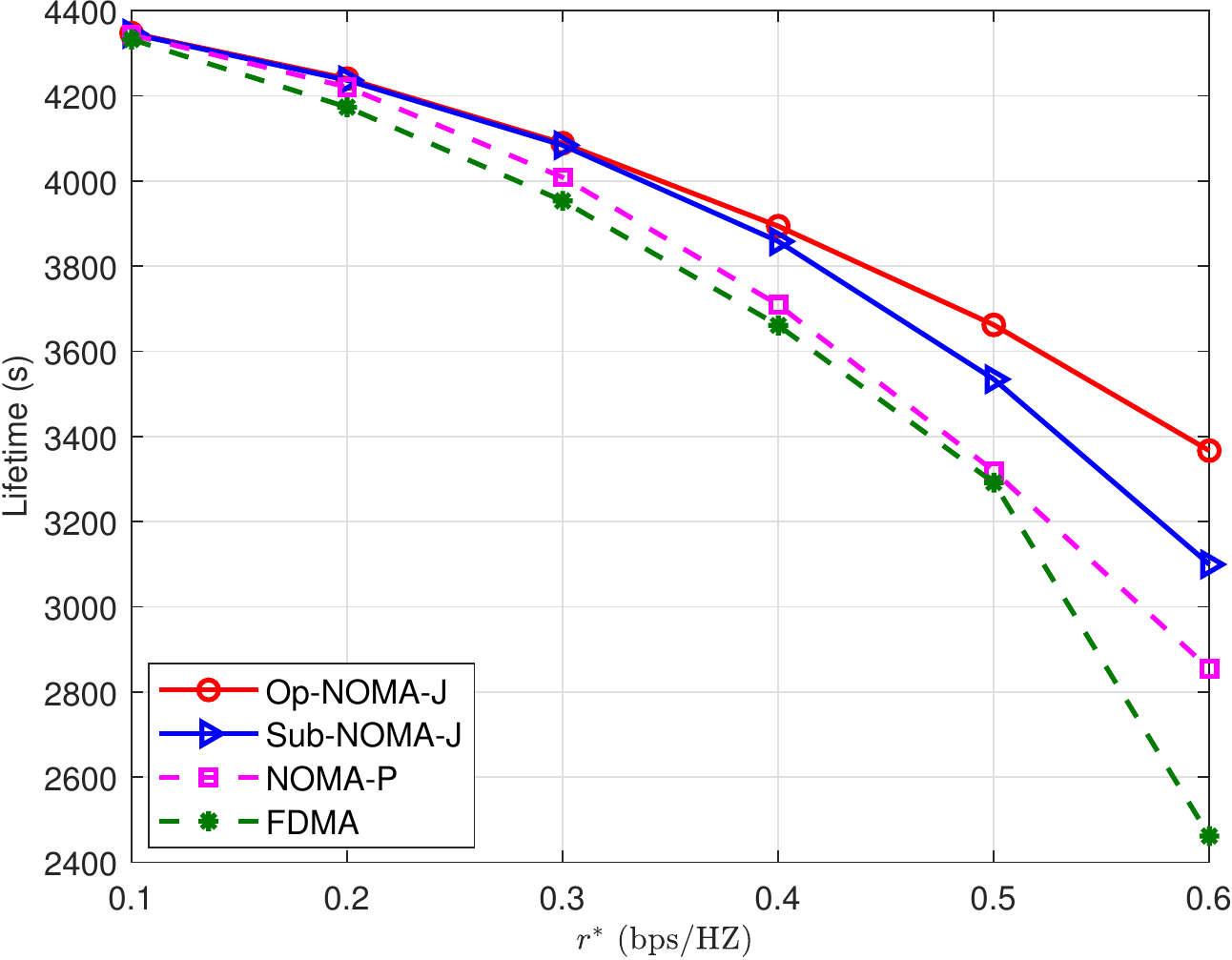}
	\caption{Lifetime versus the QoS requirement $r^*$ for different algorithms.}\label{fig: Compare_value_rate}
\end{figure}

\begin{figure}[!t]
	\centering
	\includegraphics[width=3in]{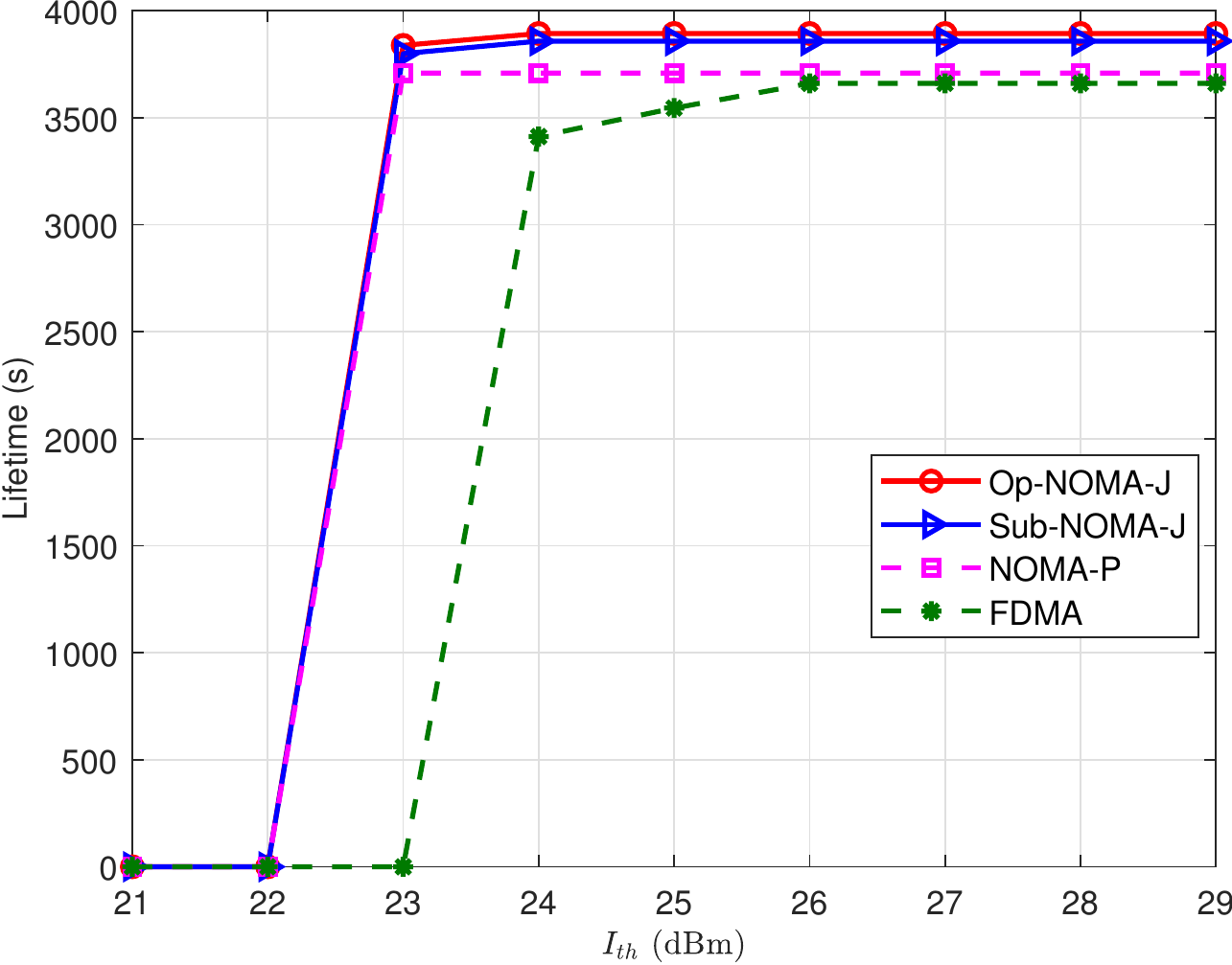}
	\caption{Lifetime versus the interference threshold $I_{th}$ for different algorithms.}\label{fig: Compare_value_threshold}
\end{figure}

Fig. \ref{fig: Compare_value_rate} shows the max-min lifetime versus the QoS requirement $r^*$ for the different algorithms.
It can be observed that the lifetime achieved by all the four algorithms rapidly decreases with the growth of $r^*$.
This is expected since as $r^*$ increases, all devices need to consume more transmit power to meet the QoS requirement, which
degrades the lifetime performance.
It can also be seen that the lifetime of the three NOMA-based algorithms significantly outperforms that of FDMA, especially when $r^*$ is large.
This result indicates that NOMA is more effective than FDMA in improving the lifetime.
Besides, the performance gap between NOMA-J (including Sub-NOMA-J and Op-NOMA-J ) and NOMA-P becomes larger as $r^*$ increases, which demonstrates the importance of the joint optimization of the UAV location, transmit power, and decoding order in enhancing lifetime performance.
\begin{figure}[!t]
	\centering
	\includegraphics[width=3in]{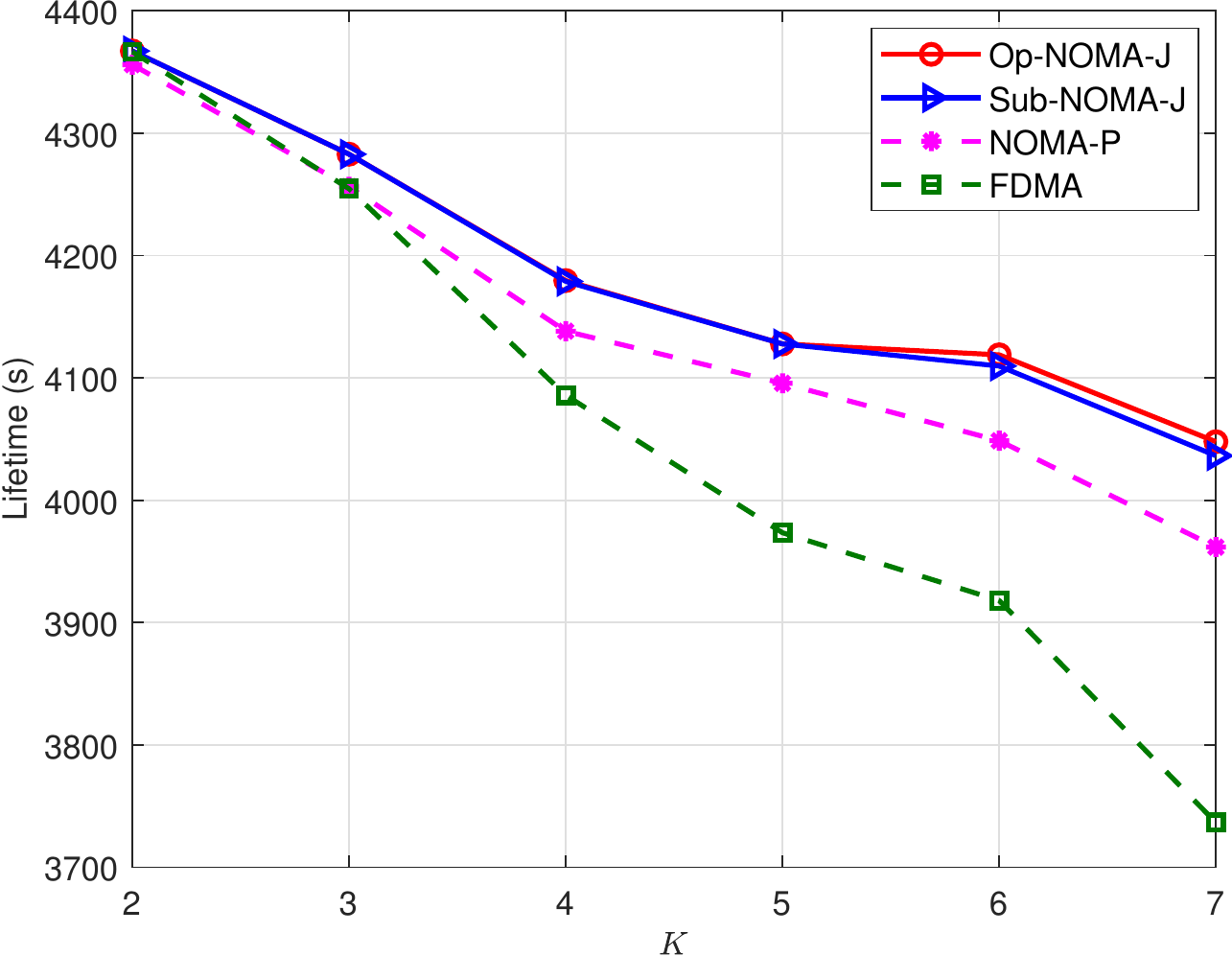}
	\caption{Lifetime versus the device number $K$ for different algorithms.}\label{fig: Compare_value_userNo}
\end{figure}


Fig. \ref{fig: Compare_value_threshold} shows the max-min lifetime versus the interference threshold $I_{th}$ for the different algorithms with $r^*=0.4$ bps/Hz.
It can be observed that when $I_{th}$ is very small (e.g., $I_{th} \leq 22$ dBm for the three NOMA-based algorithms and $I_{th} \leq 23$ dBm for FDMA), the lifetime approaches zero.
This is because there exists at least one IoT device whose transmit power cannot satisfy strict probabilistic interference constraints.
Furthermore, as $I_{th}$ increases, the lifetime of all algorithms gradually increases before getting saturation at a sufficiently large $I_{th}$.
The reason is that when $I_{th}$ is moderate, the available transmit power increases with the relaxation of interference constraints, which provides higher flexibility in power control to achieve a better max-min lifetime.
However, when $I_{th}$ is sufficiently large, the interference constraints become loose while the power constraints are active by the QoS requirement and the maximum transmit power.
Moreover, our proposed NOMA-J outperforms the other algorithms, which also
validates the necessity of joint optimization of the UAV location, transmit power, and decoding order.

Fig. \ref{fig: Compare_value_userNo} shows the max-min lifetime versus the device number $K$ for the different algorithms with $r^*=0.4$ bps/Hz.
It can be observed that as $K$ increases, the performance of FDMA sharply degrades while the three NOMA-based algorithms can maintain acceptable performance.
In addition, the lifetime of  FDMA is superior to that of NOMA-P when $K = 2 $.
However, with the growth of $K$, the performance of  NOMA-P will surpass that of  FDMA, and the lifetime gap between NOMA-P and FDMA becomes larger.
From the figure, the proposed Op-NOMA-J always achieves the highest lifetime, especially when $K$ is large, which further confirms the performance gain brought by NOMA.
Moreover, the performance of Sub-NOMA-J is close to that of Op-NOMA-J when $K \leq 5$, and it can still reach more than $99.6\%$ of the optimal performance when $K=7$.

\section{Conclusion}\label{sec: conclusion}
In this paper, we study the cognitive NOMA uplink transmission in UAV-enabled IoT networks.
Specifically, the minimum lifetime of IoT devices is maximized via jointly optimizing the UAV location, transmit power, and decoding order without violating the QoS and interference constraints with imperfect CSI.
Although the formulated problem is non-convex, we solve it optimally via the Lagrange-duality-based algorithm, which entails high computational complexity.
Furthermore, by equivalently transforming the original problem into a more tractable form, we propose a low-complexity iterative algorithm by leveraging the SCA technique and penalty method, which offers a sub-optimal solution.
Moreover, we also unify two existing decoding order formulations.
Numerical results demonstrate the effectiveness of joint UAV location, transmit power, and decoding order optimization.
In practice, it is more preferable to use the Lagrange-duality-based algorithm in small-scale networks and to apply the iterative algorithm in large-scale networks.
Moreover, future research could focus on extending the problem to more practical scenarios, such as multi-UAV cooperative networks and massive MIMO UAV networks.

\appendices
\section{Proof of Lemma 1}\label{appendix:lemma1}
Define $\beta_m \triangleq p_{m}{h}_{m}$. According to \eqref{equ: power3}, we have
\begin{align}
	\beta_m = \left( 2^{r^*}-1 \right) \left( \sum_{n=m+1}^{K}\beta_n+\sigma^2\right).\label{equ: xequation}
\end{align}
Hence,
\begin{align}
	\beta_m-\left( 2^{r^*}-1 \right)\beta_{m+1}&=\left( 2^{r^*}-1 \right)\left( \sum_{n=m+2}^{K}\beta_n+\sigma^2\right)\nonumber\\
	&\overset{(a)}{=}\beta_{m+1},
\end{align}
where $(a)$ is due to the definition of $\beta_{m+1}$.
Then, we have
\begin{align}
	\beta_m = 2^{r^*}\beta_{m+1}.
\end{align}
Therefore, $\{\beta_m\}_{m=1}^{K}$ forms a proportional sequence with a common ratio $2^{-r^*}$. By referring to the fact that $\beta_K=(2^{r^*}-1)\sigma^2$, we have
\begin{align}
	\beta_m = (2^{r^*}-1)\sigma^2 2^{(K-m)r^*}. \label{equ: xsolution}
\end{align}
Consequently, we have
\begin{align}
	p_{m} = \frac{(2^{r^*}-1)\sigma^2}{{h}_{m}}2^{(K-m)r^*}=c_m(H^2+||\mathbf{q}-w_{m}||^2),\forall m,\nonumber
\end{align}
where $c_m\triangleq \frac{(2^{r^*}-1)\sigma^2}{\rho_0}2^{(K-m)r^*}$.
\hfill $\blacksquare$

\section{Proof of Proposition \ref{lem:lemma3}}\label{appendix:lemma2}
To prove proposition \ref{lem:lemma3}, we need the following Lemmas.
We first define an intermediate variable $X_k\triangleq \sum_{j=1}^{K}\alpha_{k,j}$, whose property is given in Lemma \ref{lem:lemma01}.
\begin{lemma} \label{lem:lemma01}
	For any $k\neq i$, we have $|X_k- X_i|\neq 0.$
\end{lemma}
\begin{proof}
	First, for $\forall k \neq i$, we have
	\begin{align}
		X_{k}-X_{i}&=\sum_{j=1}^{K}\alpha_{k, j}- \sum_{j=1}^{K}\alpha_{i, j}\nonumber \\
		=&\alpha_{k, i}-\alpha_{i, k}+\sum_{j=1,j \neq k, i}^{K}\left(\alpha_{k, j}-\alpha_{i, j}\right),
		\label{equ: tmp12}
	\end{align}
	Due to constraints
	\eqref{equ: SIC} and \eqref{equ: SICpair}, we must have $\alpha_{k,i}\neq \alpha_{i, k}$. We will discuss the following two cases.
	\paragraph{$\alpha_{k,i}=0,\alpha_{i, k}=1$}Referring to the fact that $\alpha_{i, j}=1-\alpha_{j, i}$ and constraint \eqref{equ: SICnew}, we have
	\begin{align}
		\alpha_{k, j}-\alpha_{i, j}=\alpha_{k, j}+\alpha_{j, i}-1\leq \alpha_{k, i}=0
	\end{align}
	Notice that in this case, $\alpha_{k, i}-\alpha_{i, k}=-1$, we have
	$X_{k}-X_{i}\leq -1$.
	\paragraph{$\alpha_{k,i}=1,\alpha_{i, k}=0$} Since constraints \eqref{equ: SICnew} holds for $\forall k, i, j$, by exchanging the index of $i$ and $j$ in \eqref{equ: SICnew}, we have
	\begin{align}
		\alpha_{k, i}+\alpha_{i,j}-1\leq \alpha_{k,j}.
	\end{align}
	This leads to
	\begin{align}
		\alpha_{k,j}-\alpha_{i,j}\geq \alpha_{k, i}-1=
		-\alpha_{i, k}=0.
	\end{align}
	In this case, $\alpha_{k, i}-\alpha_{i, k}=1$ holds.
	Thus, we have $X_{k}-X_{i} \geq 1$.
	
	Based on the above two cases, we conclude that
	\begin{align}
		\big|X_{k}-X_{i} \big| \geq 1,\forall k \neq i,
	\end{align}
	which completes the proof.
	\hfill $\blacksquare$
\end{proof}
\begin{lemma} \label{lem:lemma02}
	Define $\alpha^{s}_{f(k),f(j)}=\alpha_{k, j}$ with $f(k)\triangleq K-\sum_{j=1}^{K}{\alpha}_{k,j}$. If $\bm \alpha$ satisfied constraints \eqref{equ: SIC}-\eqref{equ: SICnew}, then we have
	\begin{align}
		{\alpha}^{s}_{f(k), f(j)}=\left\{\begin{array}{ll}{1,} & {\forall f(k)<f(j),} \\ {0,} & {otherwise. } \\ \end{array}\right. \forall k,j.
	\end{align}
\end{lemma}	
\begin{proof}
	Note that $X_{k} \in \{0,1,\cdots,K-1\}$ since $\alpha_{k,j} \in \{0,1\}$.
	By Lemma \ref{lem:lemma01}, we have $|X_k- X_i|\neq 0$, $\forall k\neq i$, then one can easily see that elements in $\{{X}_k\}_{k=1}^K$ correspond to $\{0,1,\cdots,K-1\}$ in a one-to-one manner.
	Therefore, by mapping $k$ to $f(k)$, we can establish a new set $\mathbf{X}^s\triangleq \{{X}^s_{f(k)}\}_{f(k)=1}^K$ with
	\begin{align}
		X_{f(k)}^\text{s}=K-f(k). \label{equ: temp01}
	\end{align}
	Note that $X_k=X^s_{f(k)}$, we have $X_k=X^s_{f(k)}=K-f(k)$, which also determines the definition of $f(k)$, i.e., $f(k)=K-X_k=K-\sum_{j=1}^{K}\alpha_{k,j}$.
	
	
	Then, by letting $\alpha^\text{s}_{f(k),f(j)}=\alpha_{k,j}$,  we can obtain a new sequence $\bm\alpha^\text{s}\triangleq\{\alpha^\text{s}_{f(k),f(j)}, k,j \in \mathcal{K}\}$.	
	In the following, we will give the closed-form expression of $\bm \alpha^s$. Note that $\bm\alpha^\text{s}$ also meets constraints \eqref{equ: SIC}-\eqref{equ: SICnew} and
	\begin{align}
		X_{f(k)}^\text{s}=X_{k}=\sum_{j=1}^{K}\alpha_{k,j}=\sum_{j=1}^{K}\alpha^\text{s}_{f(k),f(j)}=\sum_{f(j)=1}^{K}\alpha^\text{s}_{f(k),f(j)}.\label{equ: temp03}
	\end{align}
Combining \eqref{equ: temp01} and \eqref{equ: temp03}, we obtain 	
\begin{align}
	\sum_{f(j)=1}^{K}\alpha^\text{s}_{f(k),f(j)}=K-f(k), \forall k. \label{equ: temp02}
\end{align}
Then, when $f(k)=1$, we have $\sum_{f(j) \neq 1}^{K}{\alpha}^{\text{s}}_{1,f(j)}=K-1$ owing to ${\alpha}^{\text{s}}_{1,1}=0$, which means
\begin{align}
	{\alpha}^{\text{s}}_{1,f(j)}=1,\forall f(j) >1.\label{equ: tmp14}
\end{align}
Substituting \eqref{equ: tmp14} into constraints \eqref{equ: SICnew}, we have
\begin{align}
	{\alpha}^{\text{s}}_{f(j),1}=0,\forall f(j) > 1. \label{equ: temp04}
\end{align}
Subsequently, when $f(k)=2$, we have $\sum_{f(j) \neq 2}^{K}{\alpha}^{\text{s}}_{2,f(j)}=K-2$. Owing to 	${\alpha}^{\text{s}}_{2,2}=0$ and ${\alpha}^{\text{s}}_{2,1}=0$ from equation \eqref{equ: temp04}, we have
\begin{align}
	{\alpha}^{\text{s}}_{2,f(j)}=1,\forall f(j) >2. \label{equ: tmp15}
\end{align}
Similarly, we have
\begin{align}
	{\alpha}^{\text{s}}_{f(j),2}=0,\forall f(j) \geq 2.
\end{align}
By successively doing the same operations for $f(k) \geq 3$, we have
\begin{align}
	{\alpha}^{\text{s}}_{f(k),f(j)}= 1, \forall f(k)<f(j).\label{equ: tmp16}
\end{align}	
Substituting \eqref{equ: tmp16} into \eqref{equ: SICnew}, we obtain
\begin{align}
	{\alpha}^{\text{s}}_{f(k),f(j)}=\left\{\begin{array}{ll}{1,} & {\text { if } f(k)<f(j),} \\ {0,} & {otherwise,} \\ \end{array}\right. \forall k,j, \label{equ: temp14}
\end{align}	
which completes the proof.	\hfill $\blacksquare$
\end{proof}

Based on Lemma \ref{lem:lemma01} and \ref{lem:lemma02}, we start to prove Proposition \ref{lem:lemma3}. Define
$ {p}^{\text{s}}_{f(k)}={p}_k$,
${h}^{\text{s}}_{f(k)}={h}_{k}$, $\forall k$. Constraints \eqref{equ: qosnew} can be rewritten as
\begin{align}
r^*& \leq \log_2\left(1+\frac{p_{k}h_{k}}{\sum_{j=1,j\neq k}^K {\alpha}_{k,j}{p}_jh_{j}+\sigma^2} \right)\nonumber\\
&= \log_2\left(1+\frac{ p_{k}h_{k}}{\sum_{j=1, f(j)\neq f(k)}^K {\alpha}^{\text{s}}_{f(k),f(j)}{p}^{\text{s}}_{f(j)} {h}^{\text{s}}_{f(j)}+\sigma^2} \right) \nonumber \\
&= \log_2\left(1+\frac{ p_{k}h_{k}}{\sum_{ f(j)= f(k)+1}^K{p}^{\text{s}}_{f(j)} {h}^{\text{s}}_{f(j)}+\sigma^2} \right)\nonumber\\
&= \log_2\left(1+\frac{ p_{k}h_{k}}{\sum_{f(j)= f(k)+1}^K {p}_{j} {h}_{j}+\sigma^2} \right).
\end{align}
\hfill $\blacksquare$
\section{Proof of Proposition \ref{lem:lemma4}}\label{appendix:lemma3}	
We will show that $\{\bm{\pi}',\mathbf p'(\bm \pi),\mathbf q\}$ is a feasible solution to problem \eqref{equ: problem33} by proving that it meets constraints
\eqref{equ: qos0}, \eqref{equ: channelsort} and \eqref{equ: power33}. To this end, set ${\pi}'(m)=f^{-1}(m)=k$ and $\pi'(n)=f^{-1}(n)=j$. Since $\{\bm \alpha, \mathbf p, \mathbf q\}$ is feasible to problem \eqref{equ: newproblem},
we have 
\begin{align}
	r^*
	& \leq \log_2\left(1+\frac{ p_{k}h_{k}}{\sum_{f(j)= f(k)+1}^K {p}_{j} {h}_{j}+\sigma^2} \right) \nonumber \\	
	&=\log_2\left(1+\frac{p'_{\pi'(m)}{h}'_{\pi'(m)}}{\sum_{n=m+1}^{K}p'_{\pi'(n)}{h'}_{\pi'(n)}+\sigma^{2}} \right),\\
	&= \log_2\left(1+\frac{ p_{k}h_{k}}{\sum_{f(j)= f(k)+1}^K {p}_{j} {h}_{j}+\sigma^2} \right)
\end{align}
where we define $h'_{\pi'(m)}=h_k$.
Thus, $\{\bm{\pi}',\mathbf p',\mathbf q\}$ meets constraints \eqref{equ: qos0}.

Secondly, we will consider constraints \eqref{equ: channelsort}.
Define ${d}^{\text{s}}_{f(k)}={d}_{k}$.
Recall the fact that
\begin{align}
	\alpha_{k,j}=\left\{\begin{array}{ll}{0,} & {\text { if } {d}_{k} > {d}_{j},} \\ {1,} & {\text { if } {d}_{k} < {d}_{j},}\\ {0 \text{ or } 1,} & {\text { if } {d}_{k} = {d}_{j},} \end{array}\right. \quad \forall k \neq j,\label{equ: tmp21}
\end{align}
which indicates
\begin{align}
	(2{\alpha}_{k,j}-1)({d}_{k} -{d}_{j})+\big|{d}_{k} -{d}_{j}\big|=0.\label{equ: temp08}
\end{align}
Then, we obtain
\begin{align}
	(2{\alpha}^{\text{s}}_{f(k),f(j)}-1)({d}^{\text{s}}_{f(k)} -{d}^{\text{s}}_{f(j)})+\big|{d}^{\text{s}}_{f(k)} -{d}^{\text{s}}_{f(j)}\big|=0.\label{equ: temp21}
\end{align}
If $f(k)<f(j)$, we can find ${\alpha}^{\text{s}}_{f(k),f(j)}=1$ from Lemma \ref{lem:lemma02}, which leads to
\begin{align}
	{d}^{\text{s}}_{f(k)} -{d}^{\text{s}}_{f(j)}+\big|{d}^{\text{s}}_{f(k)} -{d}^{\text{s}}_{f(j)}\big|=0, \forall f(k)<f(j),
\end{align}
or equivalently,
\begin{align}\label{equ: tmp23}
	{d}^{\text{s}}_{f(k)} \leq {d}^{\text{s}}_{f(j)}, \forall f(k)<f(j).
\end{align}
In other words,
\begin{align}
	{d}_{k} \leq {d}_{j}, \forall f(k)<f(j).\label{equ: tmp13}
\end{align}
Furthermore, referring to the fact that $d'_{\pi'(m)}=d_k$ and $d'_{\pi'(n)}=d_j$, it follows
\begin{align}
	{d}'_{\pi'(m)} \leq {d}'_{\pi'(n)}, \forall m<n,
\end{align}	
which indicates that 	 $\{\bm \pi',\mathbf p',\mathbf q\}$ meets constraints \eqref{equ: channelsort}.

Thirdly, it is obvious that $\{\bm \pi',\mathbf{p}',\mathbf q\}$ satisfies constraints (12c).
In summary, $\{\bm \pi',\mathbf{p}',\mathbf q\}$ is a feasible solution to problem \eqref{equ: problem33}.
Furthermore, define $\bar L_1$ and $\bar L_2$ as the objective value of problem \eqref{equ: problem33} obtained at $\{\bm  \pi',\mathbf{p}',\mathbf q\}$ and that of problem \eqref{equ: newproblem} obtained at $\{\bm \alpha, \mathbf p, \mathbf{q}\}$.
It is obvious that $\bar L_1$ and $\bar L_2$ only depend on $\{p'_{\pi'(m)}, E_{\pi'(m)}\}$ and $\{p_k,E_k\}$, respectively. Since  the same elements are arranged in different orders for $\mathbf p$ and $\mathbf { p}'$, we have $\bar L_1=\bar L_2$. 	
$\hfill\blacksquare$

\section{Proof of Lemma 5}\label{appendix:lemma5}
We will show that $\{\tilde {\bm{\alpha}},\tilde {\mathbf p} , \mathbf{{q}}\}$ is a feasible solution of problem \eqref{equ: newproblem} by proving that it meets constraints \eqref{equ: newpower}-\eqref{equ: SICnew}.
To see this, set $\pi(m)=k$, $\pi(n)=j$ and $\pi(l)=i$. 

First, since $\{{\bm \pi}, \mathbf{p},\mathbf{q}\}$ is feasible to problem \eqref{equ: qosnew}, we have %
\begin{align}
	r^* &\leq\log_2\left(1+\frac{{p}_{\pi(m)}h_{\pi(m)}}{\sum_{n =m+1}^K {p}_{\pi(n)}h_{\pi(n)}+\sigma^2} \right)\nonumber\\
	&=\log_2\left(1+\frac{{p}_{\pi(m)}{h}_{\pi(m)}}{\sum_{m=1,m\neq n}^K \tilde{\alpha}_{\pi(m),\pi(n)}{p}_{\pi(n)}{h}_{\pi(n)}+\sigma^2} \right), \nonumber\\
	&=\log_2\left(1+\frac{\tilde p_{k}{\tilde h}_{k}}{\sum_{j=1,j\neq k}^K \tilde\alpha_{k,j}{p}_j{h}_{j}+\sigma^2} \right),
\end{align}
where we define $\tilde h_k=h_{\pi(m)}$.
Thus, $\{\tilde {\bm{\alpha}},\tilde {\mathbf{{p}}},\mathbf{{q}}\}$ meets constraints \eqref{equ: qosnew}.

Secondly, we discuss constraints \eqref{equ: SICnew} in the following four cases.
\paragraph{$m<n<l$} In this case, we have $\tilde{\alpha}_{\pi(m),\pi(n)}=1$, $\tilde{\alpha}_{\pi(n),\pi(l)}=1$ and $\tilde{\alpha}_{\pi(m),\pi(l)}=1$ such that $\tilde{\alpha}_{\pi(m),\pi(n)}+\tilde{\alpha}_{\pi(n),\pi(l)}-1 =\tilde{\alpha}_{\pi(m),\pi(l)}$.
\paragraph{$l<n<m$} In this case, we obtain $\tilde{\alpha}_{\pi(m),\pi(n)}=0$, $\tilde{\alpha}_{\pi(n),\pi(l)}=0$ and $\tilde{\alpha}_{\pi(m),\pi(l)}=0$ and then $\tilde{\alpha}_{\pi(m),\pi(n)}+\tilde{\alpha}_{\pi(n),\pi(l)}-1 <\tilde{\alpha}_{\pi(m),\pi(l)}$.
p\paragraph{$n<l,n<m$} In this case, we have $\tilde{\alpha}_{\pi(n),\pi(l)}=0$ and $\tilde{\alpha}_{\pi(n),\pi(l)}=1$ such that $\tilde{\alpha}_{\pi(m),\pi(n)}+ \tilde{\alpha}_{\pi(n),\pi(l)}-1 \leq \tilde{\alpha}_{\pi(m),\pi(l)}$ due to $\tilde{\alpha}_{\pi(m),\pi(l)} \geq 0$.
\paragraph{$n>l,n>m$ } In this case, it has $\tilde{\alpha}_{\pi(m),\pi(n)} = 1$ and $\tilde{\alpha}_{\pi(n),\pi(l)}=0$. Hence, $\tilde{\alpha}_{\pi(m),\pi(n)}+ \tilde{\alpha}_{\pi(n),\pi(l)}-1 \leq \tilde{\alpha}_{\pi(m),\pi(l)}$.

Based on the above four cases, we conclude $\{\tilde {\bm \alpha},\tilde{\mathbf{{p}}},\mathbf{{q}}\}$ meets \eqref{equ: SICnew}.

Thirdly, we come to constraints \eqref{equ: SIC}.
Since $\tilde {\bm \alpha}$ always meets constraints \eqref{equ: SIC} when $d_{\pi(m)} = 	d_{\pi(n)}$, we only need to focus on when $d_{\pi(m)} \neq	d_{\pi(n)}$ in the following.
Based on the definition of $\tilde{\alpha}_{\pi(m),\pi(n)} $, we have 
\begin{align}
	\tilde{\alpha}_{\pi(m),\pi(n)} = \frac{1}{2}\left(\frac{\big|n-m\big|}{n-m}+1\right), \forall n \neq m.\label{equ: tmp5}
\end{align}
Referring to the fact that 	$d_{\pi(1)} \leq d_{\pi(2)} \leq \cdots \leq d_{\pi(K)}$, we have
\begin{align}
	\left\{\begin{array}{ll}{d_{\pi(m)} \leq 	d_{\pi(n)},} & {\text { if } m<n,} \\ {d_{\pi(m)} \geq 	d_{\pi(n)},} & {\text { if } m>n. } \end{array}\right.
\end{align}
Thus, we have
\begin{align}
	\frac{\big|n-m\big|}{n-m}=\frac{\big|d_{\pi(n)} -	d_{\pi(m)}\big|}{d_{\pi(n)} -	d_{\pi(m)}}, \forall n, m \in\mathcal{M}_1,\label{equ: tmp6}
\end{align}
where $\mathcal{M}_1 \triangleq \{(n,m )\big| d_{\pi(m)} \neq 	d_{\pi(n)}, n \neq m\}$.
Substituting \eqref{equ: tmp6} into \eqref{equ: tmp5}, we have
\begin{align}
	\tilde{\alpha}_{\pi(m),\pi(n)} = \frac{1}{2}\left(\frac{\big|d_{\pi(n)} -	d_{\pi(m)}\big|}{d_{\pi(n)} -	d_{\pi(m)}}+1\right) ,\forall n, m \in\mathcal{M}_1 .\label{equ: tmp7}
\end{align}
This indicates
\begin{align}
	\tilde{\alpha}_{\pi(m),\pi(n)}=\left\{\begin{array}{ll}{0,} & {\text { if } d_{\pi(m)} > 	d_{\pi(n)},} \\ {1,} & {\text { if } d_{\pi(m)} <	d_{\pi(n)}.} \end{array}\right. \quad \forall n, m \in\mathcal{M}_1,
\end{align}
Hence, $\{\tilde {\bm \alpha},\tilde{\mathbf{{p}}},\mathbf{{q}}\}$ meets constraints \eqref{equ: SIC}.

Finally,
it is obvious that $\{\tilde {\bm \alpha},\tilde{\mathbf{{p}}},\mathbf{{q}}\}$ meets constraints \eqref{equ: newpower}, \eqref{equ: SICself} and \eqref{equ: SICpair}.
In summary, $\{\tilde {\bm \alpha},\tilde{\mathbf{{p}}},\mathbf{{q}}\}$ is a feasible solution to problem \eqref{equ: newproblem}.

Furthermore, define $\hat L_1$ and $\hat L_2$ as the objective value of problem \eqref{equ: problem33} obtained at $\{\bm \pi,\mathbf p,\mathbf q\}$ and that of problem \eqref{equ: newproblem} obtained at $\{\tilde {\bm \alpha},\tilde{\mathbf{{p}}},\mathbf{{q}}\}$.
For the same reason in Appendix C, we have $\hat L_1=\hat L_2$.
$\hfill\blacksquare$

\bibliographystyle{IEEEtran}
\bibliography{IEEEtran.bib}

\begin{thebibliography}{10}
\providecommand{\url}[1]{#1}
\csname url@samestyle\endcsname
\providecommand{\newblock}{\relax}
\providecommand{\bibinfo}[2]{#2}
\providecommand{\BIBentrySTDinterwordspacing}{\spaceskip=0pt\relax}
\providecommand{\BIBentryALTinterwordstretchfactor}{4}
\providecommand{\BIBentryALTinterwordspacing}{\spaceskip=\fontdimen2\font plus
\BIBentryALTinterwordstretchfactor\fontdimen3\font minus
  \fontdimen4\font\relax}
\providecommand{\BIBforeignlanguage}[2]{{%
\expandafter\ifx\csname l@#1\endcsname\relax
\typeout{** WARNING: IEEEtran.bst: No hyphenation pattern has been}%
\typeout{** loaded for the language `#1'. Using the pattern for}%
\typeout{** the default language instead.}%
\else
\language=\csname l@#1\endcsname
\fi
#2}}
\providecommand{\BIBdecl}{\relax}
\BIBdecl

\bibitem{6740844}
A.~Zanella, N.~Bui, A.~Castellani, L.~Vangelista, and M.~Zorzi, ``Internet of
  things for smart cities,'' \emph{IEEE Internet Things J.}, vol.~1, no.~1, pp.
  22--32, Feb. 2014.

\bibitem{islam2020development}
M.~M. Islam, A.~Rahaman, and M.~R. Islam, ``Development of smart healthcare
  monitoring system in {IoT} environment,'' \emph{SN computer science}, vol.~1,
  pp. 1--11, May 2020.

\bibitem{huang2015green}
X.~Huang, T.~Han, and N.~Ansari, ``On green-energy-powered cognitive radio
  networks,'' \emph{IEEE Commun. Surveys Tuts.}, vol.~17, no.~2, pp. 827--842,
  2nd Quart. 2015.

\bibitem{7559749}
X.~{Xu}, W.~{Yang}, Y.~{Cai}, and S.~{Jin}, ``On the secure spectral-energy
  efficiency tradeoff in random cognitive radio networks,'' \emph{IEEE Jour.
  Sel. Areas Commun.}, vol.~34, no.~10, pp. 2706--2722, Oct. 2016.

\bibitem{maraqa2020survey}
O.~Maraqa, A.~S. Rajasekaran, S.~Al-Ahmadi, H.~Yanikomeroglu, and S.~M. Sait,
  ``A survey of rate-optimal power domain {NOMA} with enabling technologies of
  future wireless networks,'' \emph{IEEE Commun. Surveys Tuts.}, vol.~22,
  no.~4, pp. 2192--2235, Sep. 2020.

\bibitem{8357810}
L.~Dai, B.~Wang, Z.~Ding, Z.~Wang, S.~Chen, and L.~Hanzo, ``A survey of
  non-orthogonal multiple access for {5G},'' \emph{IEEE Commun. Surveys Tuts.},
  vol.~20, no.~3, pp. 2294--2323, May 2018.

\bibitem{8438896}
Q.~{Wu} and R.~{Zhang}, ``Common throughput maximization in {UAV}-enabled
  {OFDMA} systems with delay consideration,'' \emph{IEEE Trans. Commun.},
  vol.~66, no.~12, pp. 6614--6627, Dec. 2018.

\bibitem{zhan2017energy}
C.~Zhan, Y.~Zeng, and R.~Zhang, ``Energy-efficient data collection in {UAV}
  enabled wireless sensor network,'' \emph{IEEE Wireless Commun. Lett.},
  vol.~7, no.~3, pp. 328--331, Jun. 2017.

\bibitem{9043712}
Y.~{Cai}, Z.~{Wei}, R.~{Li}, D.~W.~K. {Ng}, and J.~{Yuan}, ``Joint trajectory
  and resource allocation design for energy-efficient secure {UAV}
  communication systems,'' \emph{IEEE Trans. Commun.}, vol.~68, no.~7, pp.
  4536--4553, Mar. 2020.

\bibitem{9522072}
D.-H. Tran, V.-D. Nguyen, S.~Chatzinotas, T.~X. Vu, and B.~Ottersten, ``{UAV}
  relay-assisted emergency communications in {IoT} networks: Resource
  allocation and trajectory optimization,'' \emph{IEEE Trans. Wireless
  Commun.}, pp. 1--1, 2021.

\bibitem{9524328}
G.~Zhang, X.~Ou, M.~Cui, Q.~Wu, S.~Ma, and W.~Chen, ``Cooperative {UAV} enabled
  relaying systems: Joint trajectory and transmit power optimization,''
  \emph{IEEE Trans. Green Commun. Netw.}, pp. 1--1, 2021.

\bibitem{9070201}
H.~Li and X.~Zhao, ``Throughput maximization with energy harvesting in
  {UAV}-assisted cognitive mobile relay networks,'' \emph{IEEE Trans. Cogn.
  Commun. Netw.}, vol.~7, no.~1, pp. 197--209, Mar. 2021.

\bibitem{9082700}
B.~Ji, Y.~Li, D.~Cao, C.~Li, S.~Mumtaz, and D.~Wang, ``Secrecy performance
  analysis of {UAV} assisted relay transmission for cognitive network with
  energy harvesting,'' \emph{IEEE Trans. Veh. Technol.}, vol.~69, no.~7, pp.
  7404--7415, Jul. 2020.

\bibitem{9364745}
P.~X. Nguyen, V.-D. Nguyen, H.~V. Nguyen, and O.-S. Shin, ``{UAV}-assisted
  secure communications in terrestrial cognitive radio networks: Joint power
  control and {3D} trajectory optimization,'' \emph{IEEE Trans. Veh. Technol.},
  vol.~70, no.~4, pp. 3298--3313, Apr. 2021.

\bibitem{8776639}
Y.~{Huang}, W.~{Mei}, J.~{Xu}, L.~{Qiu}, and R.~{Zhang}, ``Cognitive {UAV}
  communication via joint maneuver and power control,'' \emph{IEEE Trans.
  Wireless Commun.}, vol.~67, no.~11, pp. 7872--7888, Nov. 2019.

\bibitem{9233353}
B.~Liu, Y.~Wan, F.~Zhou, Q.~Wu, and R.~Q. Hu, ``Robust trajectory and
  beamforming design for cognitive {MISO} {UAV} networks,'' \emph{IEEE Wireless
  Commun. Lett.}, vol.~10, no.~2, pp. 396--400, Feb. 2021.

\bibitem{9014324}
Z.~Chu, W.~Hao, P.~Xiao, and J.~Shi, ``{UAV} assisted spectrum sharing
  ultra-reliable and low-latency communications,'' in \emph{Proc. IEEE
  GLOBECOM}, 2019, pp. 1--6.

\bibitem{9126800}
W.~Chen, S.~Zhao, R.~Zhang, Y.~Chen, and L.~Yang, ``{UAV}-assisted data
  collection with nonorthogonal multiple access,'' \emph{IEEE Internet Things
  J.}, vol.~8, no.~1, pp. 501--511, Jan. 2021.

\bibitem{9257576}
D.~Zhai, H.~Li, X.~Tang, R.~Zhang, Z.~Ding, and F.~R. Yu, ``Height optimization
  and resource allocation for {NOMA} enhanced {UAV}-aided relay networks,''
  \emph{IEEE Trans. Commun.}, vol.~69, no.~2, pp. 962--975, Feb. 2021.

\bibitem{9509753}
B.~Hu, L.~Wang, S.~Chen, J.~Cui, and L.~Chen, ``An uplink throughput
  optimization scheme for {UAV}-enabled urban emergency communications,''
  \emph{IEEE Internet Things J.}, pp. 1--1, 2021.

\bibitem{9194041}
R.~Zhang, X.~Pang, J.~Tang, Y.~Chen, N.~Zhao, and X.~Wang, ``Joint location and
  transmit power optimization for {NOMA}-{UAV} networks via updating decoding
  order,'' \emph{IEEE Wireless Commun. Lett.}, vol.~10, no.~1, pp. 136--140,
  Jan. 2021.

\bibitem{8848428}
D.~Hu, Q.~Zhang, Q.~Li, and J.~Qin, ``Joint position, decoding order, and power
  allocation optimization in {UAV}-based {NOMA} downlink communications,''
  \emph{IEEE Syst. J.}, vol.~14, no.~2, pp. 2949--2960, Jun. 2020.

\bibitem{9113466}
M.-J. Youssef, J.~Farah, C.~A. Nour, and C.~Douillard, ``Full-duplex and
  backhaul-constrained {UAV}-enabled networks using {NOMA},'' \emph{IEEE Trans.
  Veh. Technol.}, vol.~69, no.~9, pp. 9667--9681, Jun. 2020.

\bibitem{9411713}
Y.~Li, H.~Zhang, K.~Long, C.~Jiang, and M.~Guizani, ``Joint resource allocation
  and trajectory optimization with {QoS} in {UAV}-based {NOMA} wireless
  networks,'' \emph{IEEE Trans. Wireless Commun.}, pp. 1--1, 2021.

\bibitem{8918266}
R.~{Tang}, J.~{Cheng}, and Z.~{Cao}, ``Joint placement design, admission
  control, and power allocation for {NOMA}-based {UAV} systems,'' \emph{IEEE
  Wireless Commun. Lett.}, vol.~9, no.~3, pp. 385--388, Dec. 2020.

\bibitem{tang2020cognitive}
N.~Tang, H.~Tang, B.~Li, and X.~Yuan, ``Cognitive {NOMA} for {UAV}-enabled
  secure communications: Joint {3D} trajectory design and power allocation,''
  \emph{IEEE Access}, vol.~8, pp. 159\,965--159\,978, Sep. 2020.

\bibitem{8988182}
N.~{Zhao}, Y.~{Li}, S.~{Zhang}, Y.~{Chen}, W.~{Lu}, J.~{Wang}, and X.~{Wang},
  ``Security enhancement for {NOMA}-{UAV} networks,'' \emph{IEEE Trans. Veh.
  Technol.}, vol.~69, no.~4, pp. 3994--4005, Feb. 2020.

\bibitem{9080059}
W.~{Wang}, J.~{Tang}, N.~{Zhao}, X.~{Liu}, X.~Y. {Zhang}, Y.~{Chen}, and
  Y.~{Qian}, ``Joint precoding optimization for secure {SWIPT} in {UAV}-aided
  {NOMA} networks,'' \emph{IEEE Trans. Commun.}, vol.~68, no.~8, pp.
  5028--5040, Aug. 2020.

\bibitem{shakhatreh2019uavs}
H.~Shakhatreh, A.~Khreishah, and B.~Ji, ``{UAV}s to the rescue: Prolonging the
  lifetime of wireless devices under disaster situations,'' \emph{IEEE Trans.
  Green Commun. Netw.}, vol.~3, no.~4, pp. 942--954, Dec. 2019.

\bibitem{8886053}
K.~{Chen}, T.~{Chang}, and T.~{Lee}, ``Lifetime maximization for uplink
  transmission in {UAV}-enabled wireless networks,'' in \emph{Proc. IEEE WCNC},
  Apr. 2019, pp. 1--6.

\bibitem{ma2021uav}
R.~Ma, R.~Wang, G.~Liu, W.~Meng, and X.~Liu, ``{UAV}-aided cooperative data
  collection scheme for ocean monitoring networks,'' \emph{IEEE Internet Things
  J.}, vol.~8, no.~17, pp. 13\,222--13\,236, Mar. 2021.

\bibitem{zhang2019optimal}
J.~Zhang, L.~Zhu, Z.~Xiao, X.~Cao, D.~O. Wu, and X.-G. Xia, ``Optimal and
  sub-optimal uplink {NOMA}: Joint user grouping, decoding order, and power
  control,'' \emph{IEEE Wireless Commun. Lett.}, vol.~9, no.~2, pp. 254--257,
  Feb. 2019.

\bibitem{lu2020uav}
J.~Lu, Y.~Wang, T.~Liu, Z.~Zhuang, X.~Zhou, F.~Shu, and Z.~Han, ``{UAV}-enabled
  uplink non-orthogonal multiple access system: Joint deployment and power
  control,'' \emph{IEEE Trans. Veh. Technol.}, vol.~69, no.~9, pp.
  10\,090--10\,102, Jun. 2020.

\bibitem{8685130}
F.~{Cui}, Y.~{Cai}, Z.~{Qin}, M.~{Zhao}, and G.~Y. {Li}, ``Multiple access for
  mobile-{UAV} enabled networks: Joint trajectory design and resource
  allocation,'' \emph{IEEE Trans. Wireless Commun.}, vol.~67, no.~7, pp.
  4980--4994, Apr. 2019.

\bibitem{xu2020joint}
Y.~Xu, T.~Zhang, D.~Yang, Y.~Liu, and M.~Tao, ``Joint resource and trajectory
  optimization for security in {UAV}-assisted {MEC} systems,'' \emph{IEEE
  Trans. Commun.}, vol.~69, no.~1, pp. 573--588, Jan. 2021.

\bibitem{nguyen2018novel}
T.~M. Nguyen, W.~Ajib, and C.~Assi, ``A novel cooperative non-orthogonal
  multiple access ({NOMA}) in wireless backhaul two-tier hetnets,'' \emph{IEEE
  Trans. Wireless Commun.}, vol.~17, no.~7, pp. 4873--4887, Jul. 2018.

\bibitem{9485092}
H.~Tang, Q.~Wu, W.~Chen, J.~Wang, and B.~Li, ``Mitigating the doubly near–far
  effect in {UAV}-enabled {WPCN},'' \emph{IEEE Trans. Veh. Technol.}, vol.~70,
  no.~8, pp. 8349--8354, Jul. 2021.

\bibitem{ali2016dynamic}
M.~S. Ali, H.~Tabassum, and E.~Hossain, ``Dynamic user clustering and power
  allocation for uplink and downlink non-orthogonal multiple access ({NOMA})
  systems,'' \emph{IEEE Access}, vol.~4, pp. 6325--6343, Aug. 2016.

\bibitem{5336868}
L.~{Musavian} and S.~{Aissa}, ``Fundamental capacity limits of cognitive radio
  in fading environments with imperfect channel information,'' \emph{IEEE
  Trans. Commun.}, vol.~57, no.~11, pp. 3472--3480, Nov. 2009.

\end{thebibliography}

\end{document}